\documentclass[times,twocolumn,final]{elsarticle}

\usepackage{framed,multirow}

\usepackage[dvipsnames]{xcolor}
\usepackage{color}

\usepackage[pdftitle={Rigid Motion Invariant Statistical Shape Modeling based on Discrete Fundamental Forms}]{hyperref}

\journal{Medical Image Analysis}

\biboptions{authoryear}

\let\cite\citep

\usepackage[utf8]{inputenc} %
\usepackage[T1]{fontenc}    %
\usepackage{hyperref}       %
\usepackage{url}            %
\usepackage{booktabs}       %
\usepackage{amsfonts}       %
\usepackage{nicefrac}       %
\usepackage{microtype}      %
\usepackage{graphicx}
\usepackage[listofformat=subsimple]{subfig}

\usepackage{makecell}

\usepackage{enumitem}

\usepackage{amsmath}
\usepackage{amssymb}
\usepackage{amsthm}
\usepackage{mathdots}

\usepackage{multicol}
\usepackage{multirow}

\usepackage{float}

\usepackage[colorinlistoftodos]{todonotes}

\newcolumntype{L}[1]{>{\raggedright\arraybackslash}p{#1}}
\newcolumntype{R}[1]{>{\raggedleft\arraybackslash}p{#1}}
\newcolumntype{C}[1]{>{\centering\arraybackslash}p{#1}}

\DeclareMathOperator{\tr}{tr}

\newcommand{\norm}[1]{\left\lVert#1\right\rVert}
\newcommand{\abs}[1]{\left\lvert#1\right\rvert}
\newcommand{\scal}[2]{\left\langle#1, #2\right\rangle}

\DeclareMathOperator*{\argmax}{arg\,max}
\DeclareMathOperator*{\argmin}{arg\,min}

\newcommand{\reference}[1]{\bar{#1}}
\newcommand{\matop}{}

\newcommand{\SO}[1][3]{\text{SO}(#1)}

\newcommand{\Sym}[1][3]{\text{Sym}^+(#1)}

\newcommand{\bbbbr}{\mathbb{R}}

\newcommand{\revisedIII}[1]{{#1}}
\newcommand{\revisedII}[1]{#1}
\newcommand{\revised}[1]{#1}

\usepackage{mathtools, cuted}
\usepackage{tikz}
\usetikzlibrary{3d,calc,intersections}

\usetikzlibrary{arrows}
\input{arrowsnew}

\usepackage{siunitx}

\newcommand{\tikzAngleOfLine}{\tikz@AngleOfLine}
\def\tikz@AngleOfLine(#1)(#2)#3{%
  \pgfmathanglebetweenpoints{%
    \pgfpointanchor{#1}{center}}{%
    \pgfpointanchor{#2}{center}}
  \pgfmathsetmacro{#3}{\pgfmathresult}%
}  

\usetikzlibrary{decorations.markings}

\tikzset{test/.style n args={3}{
    postaction={
    decorate,
    decoration={
    markings,
    mark=between positions 0 and \pgfdecoratedpathlength step 0.2pt with {
    \pgfmathsetmacro\myval{multiply(
        divide(
        \pgfkeysvalueof{/pgf/decoration/mark info/distance from start}, \pgfdecoratedpathlength
        ),
        100
    )};
    \pgfsetfillcolor{#3!\myval!#2};
    \pgfpathcircle{\pgfpointorigin}{#1};
    \pgfusepath{fill};}
}}}}

\tikzset{test123/.style n args={3}{
    postaction={
    decorate,
    decoration={
    markings,
    mark=between positions 0 and \pgfdecoratedpathlength step 0.2pt with {
    \pgfmathsetmacro\myval{multiply(
        divide(
        \pgfkeysvalueof{/pgf/decoration/mark info/distance from start}, \pgfdecoratedpathlength
        ),
        100
    )};
    \pgfsetfillcolor{#3!\myval!#2};
    \pgfpathcircle{\pgfpointorigin}{#1};
    \pgfusepath{fill};}
}}}}

\definecolor{backtriangleA}{HTML}{F2F2F6}  
\definecolor{backtriangleB}{HTML}{F7F7FB}  
\definecolor{backtriangleC}{HTML}{E5E5E9}  
\definecolor{backtriangleD}{HTML}{C8C8CC}  
\definecolor{backtriangleE}{HTML}{D4D4D8}  
\definecolor{backtriangleF}{HTML}{AFAFB1}  
\definecolor{backtriangleG}{HTML}{A7A7AA}  
\definecolor{backtriangleH}{HTML}{A2A2A4}  
\definecolor{backtriangleI}{HTML}{A1A1A4}
\definecolor{backtriangleJ}{HTML}{FEFEFF}

\definecolor{backtriangle}{RGB}{252,252,255}  
\definecolor{fronttriangle}{HTML}{FCE0F1} %

\definecolor{c6}{HTML}{D1E9F8}
\definecolor{c5}{HTML}{D9ECF9}
\definecolor{c4}{HTML}{E1F0FB}
\definecolor{c3}{HTML}{EAF4FC}
\definecolor{c2}{HTML}{F3F8FD}
\definecolor{c1}{HTML}{FCFCFF}

\definecolor{cc6}{HTML}{FFFFFF}
\definecolor{cc5}{HTML}{FEF9FC}
\definecolor{cc4}{HTML}{FEF3FA}
\definecolor{cc3}{HTML}{FDEDF7}
\definecolor{cc2}{HTML}{FDE7F4}
\definecolor{cc1}{HTML}{FCE0F1}

\usepackage{stackrel}

\usepackage{siunitx}
\makeatletter
\newcommand{\gettikzxy}[3]{%
  \tikz@scan@one@point\pgfutil@firstofone#1\relax
  \edef#2{\the\pgf@x}%
  \edef#3{\the\pgf@y}%
}
\makeatother

\makeatletter
\def\convertto#1#2{\strip@pt\dimexpr #2*65536/\number\dimexpr 1#1}
\makeatother

\newlength{\heightOfB}
\DeclareRobustCommand{\ShowColormap}{\raisebox{-0.14em}{\includegraphics[height=\heightOfB]{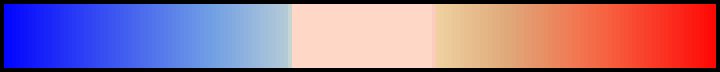}}}
\DeclareRobustCommand{\ShowColormapOAI}{\raisebox{-0.14em}{\includegraphics[height=\heightOfB]{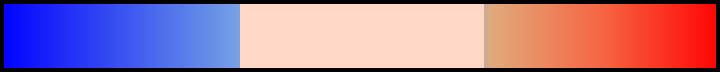}}}

\usepackage{overpic}

\begin{document}

\begin{frontmatter}
\title{Rigid Motion Invariant Statistical Shape Modeling based on Discrete Fundamental Forms
\\\vspace{10pt}\normalsize{{Data from the Osteoarthritis Initiative and the Alzheimer's Disease Neuroimaging Initiative}}
}

\author[ZIB]{Felix Ambellan\corref{mycorrespondingauthor}}
\ead{ambellan@zib.de, orcid 0000-0001-9415-0859}

\author[ZIB,1KS]{Stefan Zachow}
\ead{zachow@zib.de, orcid 0000-0001-7964-3049}

\author[FU]{Christoph von Tycowicz}
\ead{vontycowicz@zib.de, orcid 0000-0002-1447-4069}

\address[ZIB]{Zuse Institute Berlin, Berlin, Germany}
\address[1KS]{1000shapes GmbH, Berlin, Germany}
\address[FU]{Freie Universit\"at Berlin, Berlin, Germany}

\cortext[mycorrespondingauthor]{Corresponding author}

\begin{abstract}
We present a novel approach for nonlinear statistical shape modeling that is invariant under Euclidean motion and thus alignment-free.
By analyzing metric distortion and curvature of shapes as elements of Lie groups in a consistent Riemannian setting, we construct a framework that reliably handles large deformations.
Due to the explicit character of Lie group operations, our non-Euclidean method is very efficient allowing for fast and numerically robust processing.
This facilitates Riemannian analysis of large shape populations accessible through longitudinal and multi-site imaging studies providing increased statistical power.
\revised{Additionally, as planar configurations form a submanifold in shape space, our representation allows for effective estimation of quasi-isometric surfaces flattenings.
We evaluate the performance of our model w.r.t.\ shape-based classification of hippocampus and femur malformations due to Alzheimer's disease and osteoarthritis, respectively.
}
\revisedII{In particular, we outperform state-of-the-art classifiers based on geometric deep learning as well as statistical shape modeling especially in presence of sparse training data.}
To provide insight into the model's ability of capturing biological shape variability, we carry out an analysis of specificity and generalization ability.
\end{abstract}
\end{frontmatter}
\section{Introduction}\label{sec:introduction}

Statistical shape models (SSMs) have become an essential tool for medical image analysis with a wide range of applications such as segmentation of anatomical structures, computer-aided diagnosis, and therapy planning.
SSMs describe the geometric variability in a population in terms of a mean shape and a hierarchy of major modes explaining the main trends of shape variation.
Based on a notion of shape space, SSMs can be learned from a database of consistently parametrized instances from the object class under study.
The resulting models provide a shape prior that can be used to constrain synthesis and analysis problems.
Moreover, their parameter space provides a compact representation that is amenable to learning algorithms (e.g. classification or clustering), evaluation, and exploration.

Standard SSMs treat the space of shapes as a Euclidean vector space allowing for linear statistics to be applied (see  e.g.~\citet{heimann2009SSMreview} and the references therein).
Linear methods, however, are often inadequate for capturing the high variability in biological shapes~\cite{davis2010regression}.
Nonlinear approaches have been developed based on geometric as well as physical concepts such as Hausdorff distance~\cite{charpiat2006distance}, elasticity~\cite{RuWi11,vonTycowicz2015averaging,zhang2015shellPCA}, and viscous flows~\cite{fuchs2009viscousMetric,brandt2016fairing,heeren2018ShellPGA}.
In general, these methods lack numerical robustness as well as fast response rates limiting their practical applicability especially in interactive applications.
To address these challenges, one line of work models shapes by a collection of primitives~\cite{fletcher2003statistics,freifeld2012LieBodies,AmbellanZachowvonTycowicz2019} %
that naturally belong to Lie groups and effectively describe local changes in shape.
Performing intrinsic calculus on the uncoupled primitives allows for fast computations while, at the same time, accounting for the nonlinearity in shape variation.
However, solving the inverse problem, i.e.\ mapping from primitives back to surface meshes, is generally non-trivial.
Recently, \citet{vonTycowicz2018efficient} presented a physically motivated approach based on differential coordinates for which the inverse problem is well-known and can be solved at linear cost.
Despite their inherent nonlinear structure, the employed representations are not invariant under Euclidean motion and, thus, analysis thereon suffers from bias due to arbitrary choices.
While the effect of rigid motions can be removed between pairs of shapes using alignment strategies, non-transitivity thereof prevents true group-wise alignment.

\revisedII{A related concept is to exploit the homogeneous structure of the ambient space and to encode displacements of points in terms of (e.g.\ rigid or affine) transformations~\cite{gilles2011frame,arsigny2003polyrigid,arsigny2009fast,mcleod2015spatio}.
Exploiting the redundancy of such representations present e.g.\ in articulated motion, these approaches provide low-dimensional encodings of deformations.
Considering the limit case of triangle-wise supported polyaffine/-rigid deformations is similar to simplicial maps underlying the construction in~\cite{freifeld2012LieBodies,vonTycowicz2018efficient,AmbellanZachowvonTycowicz2019} as well as our setup.
However, the latter employ differential characterizations of such maps that remove translational components and put local geometric variability into focus.
}

This work presents a novel shape representation based on discrete fundamental forms that is invariant under Euclidean motion.
We endow this representation with a Lie group structure that admits bi-invariant metrics and therefore allows for consistent analysis using manifold-valued statistics based on the Riemannian framework.
Furthermore, we derive a simple, efficient, robust, yet accurate (i.e.\ without resorting to model approximations) solver for the inverse problem that allows for interactive applications. \revised{Beyond statistical shape modeling the proposed framework is amenable for surface processing such as quasi-isometric flattening.} \revisedII{A publicly available implementation of the proposed model is given in the Morphomatics\footnote{\revisedIII{morphomatics.github.io}} library.}

Although in computer graphics and vision communities, rotation invariant differential coordinates have also been successfully employed for geometry processing applications, e.g.~\citet{kircher2008free,hasler2009statistical,gao2016}, these approaches fall short of a fully intrinsic treatment (e.g. due to lack of bi-invariant group structure and linearization) and have not been adapted to the field of SSMs.
\revisedII{A recent string of contributions investigates functional characterizations of intrinsic and extrinsic geometry~\cite{rustamov2013map,corman2017functional,wang2018steklov} to obtain shape descriptors.
While the underlying functional map framework alleviates the requirement on point-to-point correspondences, the reduced function spaces are based on low-frequency variations and, thus, prone to insensitivity for localized shape variability (such as osteophyte formation during the course of osteoarthritis).}

\section{Rotation Invariant Surface Representation}\label{sec:surfaceRepresentation}

In this section, we derive a discrete surface representation based on concepts from differential geometry of smooth surfaces.
\revised{This representation's key feature, its invariance under Euclidean motion and hence well-suitedness for shape analysis purposes, arises directly from discretization of surface theoretical results. Finally, the proposed representation setting exhibits a Lie group structure that we endow with a bi-invariant metric in order to ensure structural unity between Riemannian and Lie group framework} \revisedII{(see e.g.~\cite{pennec2020affineconnection}).}

\subsection{\revised{Fundamental Forms and Surface Theory}}\label{sec:relationSurfaceTheory}
To every smooth surface there uniquely exist two smoothly pointwise varying and symmetric bilinear forms on the tangent plane, the so called \emph{fundamental forms}.
The \emph{first} fundamental form $\mathrm{I}$ (a.k.a. metric tensor) is positive-definite and allows for angle, length and area measurement.
The \emph{second} fundamental form $\mathrm{I\!I}$ describes the curvature of the surface.
A prominent result in classical mathematics, the \emph{Fundamental Theorem of Surface Theory} according to Bonnet ($\approx$1860, e.g.~\citet{doCarmo1976geomCurveSurf}~Sec.~4.3), states that if given two symmetric bilinear forms (one of them positive-definite), s.t.\ for both certain integrability conditions hold (viz.\ the Gau\ss--Codazzi equations), then they (locally) determine uniquely, up to global rotation and translation, a surface embedded in three dimensional space with these two as its fundamental forms.
Therefore, a discrete description of the fundamental forms is an excellent candidate for a rotation-invariant surface representation. In the following, we will denote our proposed shape model as the \emph{fundamental coordinate model} (FCM).

\subsection{Discretization}\label{sec:discretization}
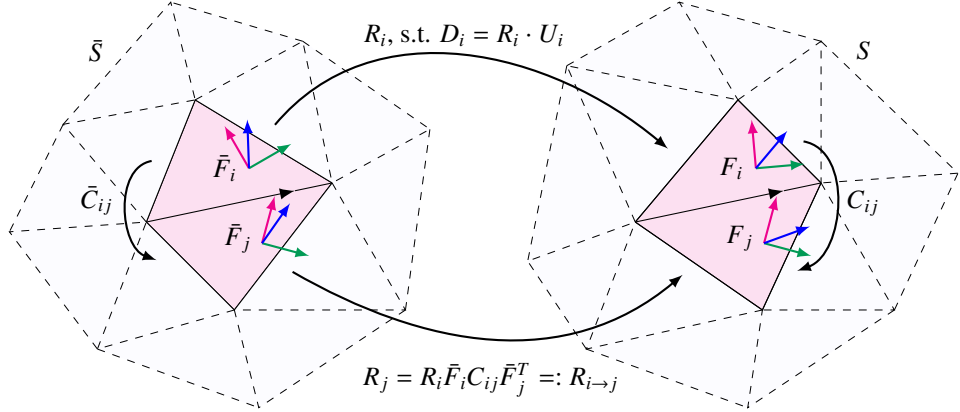
\begin{figure*}[t]
\centering
\begin{tikzpicture}[scale=0.65]

    \coordinate (A) at (0,0);
    \coordinate (B) at (1, 2.5);
    \coordinate (C) at (3.8, 0.8);
    \coordinate (D) at (1.8,-1.8);        
    
    \draw[fill=fronttriangle] (A) -- (B) -- (C);
    \draw[fill=fronttriangle] (A) -- (D) -- (C);    
    
    \coordinate (E1) at (-1.7,2.0);
    \coordinate (E2) at (3.2,3.7);
    \coordinate (E3) at (5.3,-1.8);
    \coordinate (E4) at (-1.0,-2.6);
    \coordinate (E5) at (0.5,4.5);
    \coordinate (E6) at (5.8,1.9);
    \coordinate (E7) at (2.3,-3.8);
    \coordinate (E8) at (-2.8,-0.2);
    
    \fill[backtriangle] (A) -- (E1) -- (B);
    \fill[backtriangle] (B) -- (E2) -- (C);
    \fill[backtriangle] (C) -- (E3) -- (D);
    \fill[backtriangle] (D) -- (E4) -- (A);
        
    \fill[backtriangle] (E1) -- (E5) -- (B);
    \fill[backtriangle] (E2) -- (E5) --(B);
    \fill[backtriangle] (E2) -- (E6) -- (C);
    \fill[backtriangle] (E3) -- (E6) -- (C);
    \fill[backtriangle] (E3) -- (E7) -- (D);
    \fill[backtriangle] (E4) -- (E7) -- (D);
    \fill[backtriangle] (E4) -- (E8) -- (A);
    \fill[backtriangle] (E1) -- (E8) -- (A);
    
    \draw[dashed, very thin] (A) -- (E1) -- (B);
    \draw[dashed, very thin] (B) -- (E2) -- (C);
    \draw[dashed, very thin] (C) -- (E3) -- (D);
    \draw[dashed, very thin] (D) -- (E4) -- (A);
    
    \draw[dashed, very thin] (E1) -- (E5) -- (B);
    \draw[dashed, very thin] (E2) -- (E5);
    \draw[dashed, very thin] (E2) -- (E6) -- (C);
    \draw[dashed, very thin] (E3) -- (E6);
    \draw[dashed, very thin] (E3) -- (E7) -- (D);
    \draw[dashed, very thin] (E4) -- (E7);
    \draw[dashed, very thin] (E4) -- (E8) -- (A);
    \draw[dashed, very thin] (E1) -- (E8);
    
    \draw (A) -- (B) -- (C);    
    \draw (A) -- (D) -- (C);
    
    \node (e) at ($(A)!0.85!(C)$) {};
    \node (e1) at ($(A)!0.75!(C)$) {};
    \node[above] (e2) at ($(A)!0.2!(C)$) {};
    \draw[-latexnew, arrowhead=0.2cm] (A) -- (e);
    \draw (e1) -- (C);
    
    \coordinate (AA) at ([xshift=10cm]A);
    \coordinate (BB) at ([xshift=11.1cm]B);
    \coordinate (CC) at ([xshift=10cm]C);
    \coordinate (DD) at ([xshift=10.8cm]D);
    
    \coordinate (EE1) at ([xshift=10.3cm, yshift=1.2cm]E1);
    \coordinate (EE2) at ([xshift=10.6cm]E2);
    \coordinate (EE3) at ([xshift=10cm]E3);
    \coordinate (EE4) at ([xshift=10cm]E4);
    \coordinate (EE5) at ([xshift=10.4cm, yshift=-.0cm]E5);
    \coordinate (EE6) at ([xshift=10.8cm, yshift=-.9cm]E6);
    \coordinate (EE7) at ([xshift=10cm]E7);
    \coordinate (EE8) at ([xshift=10.6cm, yshift=-.6cm]E8);
    
    \fill[backtriangle] (AA) -- (EE1) -- (BB);
    \fill[backtriangle] (BB) -- (EE2) -- (CC);
    \fill[backtriangle] (CC) -- (EE3) -- (DD);
    \fill[backtriangle] (DD) -- (EE4) -- (AA);
        
    \fill[backtriangle] (EE1) -- (EE5) -- (BB);
    \fill[backtriangle] (EE2) -- (EE5) -- (BB);
    \fill[backtriangle] (EE2) -- (EE6) -- (CC);
    \fill[backtriangle] (EE3) -- (EE6) -- (CC);
    \fill[backtriangle] (EE3) -- (EE7) -- (DD);
    \fill[backtriangle] (EE4) -- (EE7) -- (DD);
    \fill[backtriangle] (EE4) -- (EE8) -- (AA);
    \fill[backtriangle] (EE1) -- (EE8) -- (AA);
    
    \draw[dashed, very thin] (AA) -- (EE1) -- (BB);
    \draw[dashed, very thin] (BB) -- (EE2) -- (CC);
    \draw[dashed, very thin] (CC) -- (EE3) -- (DD);
    \draw[dashed, very thin] (DD) -- (EE4) -- (AA);
    
    \draw[dashed, very thin] (EE1) -- (EE5) -- (BB);
    \draw[dashed, very thin] (EE2) -- (EE5);
    \draw[dashed, very thin] (EE2) -- (EE6) -- (CC);
    \draw[dashed, very thin] (EE3) -- (EE6);
    \draw[dashed, very thin] (EE3) -- (EE7) -- (DD);
    \draw[dashed, very thin] (EE4) -- (EE7);
    \draw[dashed, very thin] (EE4) -- (EE8) -- (AA);
    \draw[dashed, very thin] (EE1) -- (EE8);
    
    \draw[fill=fronttriangle] (AA) -- (BB) -- (CC);
    \draw[fill=fronttriangle] (AA) -- (DD) -- (CC);
    \draw (AA) -- (BB) -- (CC);
    \draw (AA) -- (DD) -- (CC);
    
    \node (ee) at ($(AA)!0.85!(CC)$) {};
    \node (ee1) at ($(AA)!0.75!(CC)$) {};
    \node[above] (ee2) at ($(AA)!0.3!(CC)$) {};
    \draw[-latexnew, arrowhead=0.2cm] (AA) -- (ee);
    \draw (ee1) -- (CC);
    
    \coordinate (tiBaryRef) at (barycentric cs:A=.333,B=.333,C=.333);
    \coordinate (tjBaryRef) at (barycentric cs:A=.333,D=.333,C=.333);
    \coordinate (tiBary) at (barycentric cs:AA=.333,BB=.333,CC=.333);
    \coordinate (tjBary) at (barycentric cs:AA=.333,DD=.333,CC=.333);
    \node at (tiBaryRef) {$\reference{F}_i$};
    \node at (tjBaryRef) {$\reference{F}_j$};
    \node at (tiBary) {$F_i$};
    \node at (tjBary) {$F_j$};

    \node(halfA) at ($(A)!0.5!(B)$) {};
    \node(halfB) at ($(A)!0.5!(D)$) {};
    \node(halfAA) at ($(CC)!0.5!(BB)$) {};
    \node(halfBB) at ($(CC)!0.5!(DD)$) {};
    
    \node(halfC) at ($(C)!0.7!(B)$) {};
    \node(halfCC) at ($(BB)!0.4!(AA)$) {};
    
    \node(halfD) at ($(D)!0.3!(C)$) {};
    \node(halfDD) at ($(AA)!0.6!(DD)$) {};
    \node(tiBaryRefFrame) at ([xshift=.5cm]tiBaryRef) {};
    \node(tjBaryRefFrame) at ([xshift=.5cm, yshift=-.1cm]tjBaryRef) {};
    
    \node(tiBaryFrame) at ([xshift=.5cm]tiBary) {};
    \node(tjBaryFrame) at ([xshift=.5cm, yshift=-.1cm]tjBary) {};
    
    \def\angle{30}    
    \draw[ForestGreen,-latex, thick] (tiBaryRefFrame.center) -- ++(\angle:1);
    \draw[Magenta,-latex, thick] (tiBaryRefFrame.center) -- ++(\angle+90:1);
    \draw[blue,-latex, thick] (tiBaryRefFrame.center) -- ++(\angle+62:1);
    
    \def\angle{-15}    
    \draw[ForestGreen,-latex, thick] (tjBaryRefFrame.center) -- ++(\angle:1);
    \draw[Magenta,-latex, thick] (tjBaryRefFrame.center) -- ++(\angle+90:1);
    \draw[blue,-latex, thick] (tjBaryRefFrame.center) -- ++(\angle+70:1);
    
    \def\angle{5}    
    \draw[ForestGreen,-latex, thick] (tiBaryFrame.center) -- ++(\angle:1);
    \draw[Magenta,-latex, thick] (tiBaryFrame.center) -- ++(\angle+90:1);
    \draw[blue,-latex, thick] (tiBaryFrame.center) -- ++(\angle+45:1);
    
    \def\angle{-15}    
    \draw[ForestGreen,-latex, thick] (tjBaryFrame.center) -- ++(\angle:1);
    \draw[Magenta,-latex, thick] (tjBaryFrame.center) -- ++(\angle+90:1);
    \draw[blue,-latex, thick] (tjBaryFrame.center) -- ++(\angle+35:1);

    \draw[thick, -latex]([xshift=-.2cm]halfA.west) to [out=190,in=160] node(helperh)[left]{} ([xshift=-.5cm, yshift=.12cm]halfB.west);
    \draw[thick, -latex]([xshift=.3cm]halfAA.east) to [out=350,in=20] node[right]{$C_{ij}$} ([xshift=.3cm, yshift=-.5cm]halfBB.west);
    
    \draw[thick, -latex]([xshift=.6cm]halfC.east) to [out=50,in=140] node[above]{$R_{i} \text{, s.t. } D_i=R_i\cdot U_i$} ([xshift=-.4cm]halfCC.west);
    \draw[thick, -latex]([xshift=.4cm]halfD.east) to [out=330,in=220] node(helper)[below]{} ([xshift=-.4cm]halfDD.west);
    
    \node(helperh2) at ([xshift=-.4cm, yshift=.15cm]helperh) {$\reference{C}_{ij}$};
    
    \node(helper2) at ([xshift=-.05cm, yshift=-.65cm]helper) {$R_{j}=R_i\reference{F}_i C_{ij} \reference{F}_j^{T} =: R_{i \rightarrow j}$};

\node(nameReference) at (-1, 3.5) {$\reference{S}$};
\node(nameInstance) at (14.7, 3.5) {${S}$};

\end{tikzpicture}
\caption{\revised{Relations between reference shape $\reference{S}$ (left) and shape $S=\phi(\reference{S})$, a deformation thereof (right), s.t.\ $D_i := \nabla\phi|_{\reference{T}_i}$. Note that each frame $F_i=R_i\reference{F}_i$ is defined solely on the respective triangle $T_i$ and all neighboring frames are connected across the shared edge of their underlying triangles via $F_iC_{ij}=F_j$.}} \label{fig:discretization}
\end{figure*}
We consider shapes that belong to a particular population of anatomical structures, s.t.\ each digital shape $S$ can be described as a left-acting deformation $\phi$ of a common reference shape $\reference{S}$ given as triangulated surface. %
Let deformation $\phi$ be affine on each triangle $\reference{T}_i$ of $\reference{S}$, then the deformation gradient $\nabla\phi$ 
is the $3 \times 3$ matrix of partial derivatives of $\phi$ and constant on each triangle $D_i := \nabla\phi|_{\reference{T}_i}$ (see e.g.~\citet{botsch2006deformation} for detailed expressions). \revisedII{Note that transition from deformation to deformation gradient provides invariance under translations.}
Assuming $\phi$ to be an orientation-preserving embedding of $\reference{S}$, we can decompose $D_i$ uniquely into its rotational $R_i$ and stretching $U_i$ components by means of the polar decomposition $D_i = R_i \matop U_i$.
Note that $U_i$ furnishes a complete description of the metric distortion of $\reference{T}_i$ and is defined in reference coordinates, hence invariant under  rotation of $S$.
Indeed, we can obtain a representation of the first fundamental form by restricting the stretches to the tangent plane.
To this end, we define an arbitrary but fixed element-wise field $\{\reference{F}_i\}$ of orthonormal frames
on $\reference{S}$, s.t.\ the last column of each frame is the normal of the respective element.
Then, we represent the metric in terms of reduced stretch $\tilde{U}_i := [\reference{F}^T_i \matop U_i \matop \reference{F}_i]_{3,3} = \mathrm{I}|_{\reference{T}_i}^{\nicefrac{1}{2}}$, where $[\,\cdot\,]_{3,3}$ denotes the submatrix with the third row and column removed.

As for the second fundamental form, we note that at a point $p \in S$ it is determined by the differential of the normal field $N$, viz.\ $\mathrm{I\!I}_p(v,w)= \mathrm{I}_p(-dN_p(v), w)$ for tangent vectors $v,w$.
For a triangulated surface, the differential $dN$ is supported along the edges.
In order to derive a representation thereof, we induce the frame field $\{F_i\}$ on $S$ consistent to $\{\reference{F}_i\}$ using the rotational part of the deformation gradient, i.e.\ $F_i = R_i \matop \reference{F}_i$.
This allows us to define \emph{transition rotations} $F_i\matop C_{ij}=F_j$ for each inner edge (incident to triangles $T_i,T_j$)
that fully describe the change in normal directions.
Note that, while both the frames $\{F_i\}$ and the rotations $\{R_i\}$ are equivariant, the transition rotations $\{C_{ij}\}$ are invariant under global rotations of $S$ and $\reference{S}$.
\revised{Further details hereon are depicted in Fig.~\ref{fig:discretization}.}  
\subsection{Group Structure}\label{sec:groupStructure}
In order to perform intrinsic statistical analysis, we derive a distance that is compatible with the underlying representation space.
In particular, we endow the space with a Lie group structure together with a bi-invariant Riemannian metric for which group and Riemannian notions of exponential and logarithm coincide.
This allows us to exploit closed-form expressions to perform geodesic calculus yielding simple, efficient, and numerically robust algorithms. We recommend chapter two of~\citet{alexandrino2015lie} to readers interested in deeper insight into bi-invariant metrics on Lie groups. \revised{Especially regarding} their existence and the geometric consequences thereof. 

Our shape representation consists of transition rotations $C_{ij}\in \SO$ (one per inner edge) and tangential stretches $\tilde{U}_i\in\Sym[2]$ (one per triangle), where $\SO$ is the Lie group of rotations in $\bbbbr^3$ and $\Sym[2]$ the space of symmetric and positive-definite $2\times2$ matrices.
Following the approach in~\citet{arsigny2006logEuclidean}, we equip $U,V\in\Sym[2]$ with a multiplication $U \circ V := \exp(\log(U)+\log(V))$, s.t.\ $\Sym[2]$ turns into a commutative Lie group. It now allows for a bi-invariant metric induced by the Frobenius inner product yielding distance 
$d_{\Sym[2]}(U,V) = \left\lVert\log(V)-\log(U)\right\rVert_F$. \revisedII{Note that this structure and metric do not exhibit the {\it{}swelling effect} of determinants in interpolation~\cite{goh2011nonparametric}.}
$\SO$ \revised{as a compact Lie group} also admits a bi-invariant metric induced by the Frobenius inner product with distance $d_{\SO}(Q,R) = \left\lVert\log(Q^{T}R)\right\rVert_F$,
s.t.\ we define our representation space as the product group $G := \SO^{n} \times \Sym[2]^{m}$ and $m,n$ the number of triangles and inner edges.
Finally, we define the distance of \revised{two shapes $S,T$ based on the respective group representation $s=s(S),t=t(T) \in G$} as
\begin{align}\label{eq:distance}
\nonumber
 d_{\omega}^{\revisedII{2}}(s, t) &= \dfrac{\omega^3}{\reference{A}_\mathcal{E}}\sum_{(i,j) \in \mathcal{E}} \reference{A}_{ij} \, d_{\SO}^{\revisedII{2}}\left(C^s_{ij}, C^t_{ij}\right) \,\\&+\, \dfrac{\omega}{\reference{A}}\sum_{i=1}^m \reference{A}_i \, d_{\Sym[2]}^{\revisedII{2}}\left(\tilde{U}^s_i, \tilde{U}^t_i\right),
\end{align}
where $\omega\in\bbbbr^{+}$ is a weighting factor, $\mathcal{E}$ is the set of inner edges, $\reference{A}_i$ is the area of triangle $\reference{T}_i$, $\reference{A}_{ij} = \nicefrac{1}{3}(\reference{A}_i+\reference{A}_j)$, $\reference{A}_\mathcal{E}=\sum_{(i,j) \in \mathcal{E}} \reference{A}_{ij}$, and $\reference{A} = \sum_{i=1}^m \reference{A}_i$.
\revised{The area terms hereby} provide invariance under refinement of the mesh as well as simultaneous scaling of $\reference{S}, S, T$, whereas $\omega$ allows for commensuration of the curvature and metric contributions \revised{in analogy to} the Koiter thin shell model (e.g.~\citet{Ciarlet2005}~Sec.~4.1).

\section{Shape Analysis and Processing}

\subsection{Statistical Shape Modeling}\label{sec:modelConstruction}The derived representation carries a rich non-Euclidean structure calling for manifold-valued generalizations for first and second moment statistical analysis. 
\revised{
By virtue of the bi-invariant metric, the proposed representation allows for consistent analysis within the Riemannian framework for which statistics are well-developed, while at the same time providing closed-form, group theoretic expressions for geodesic calculus~\cite{vonTycowicz2018efficient}.
In particular, we employ the \textit{Riemannian center of mass} that provides a rigorous notion of a mean $\mu$ of elements $\{s_i=s_i(S_i)\}$ and can be efficiently computed using the Gauss-Newton descent algorithm~\cite{pennec2006intrinsic, arsigny2006logEuclidean}:
\begin{align*}
    \mu^{k+1}&=\exp\left(\sum_{i} \log{\left(s_i\cdot (\mu^{k})^{-1}\right)}\right)\cdot \mu^{k}.
\end{align*}
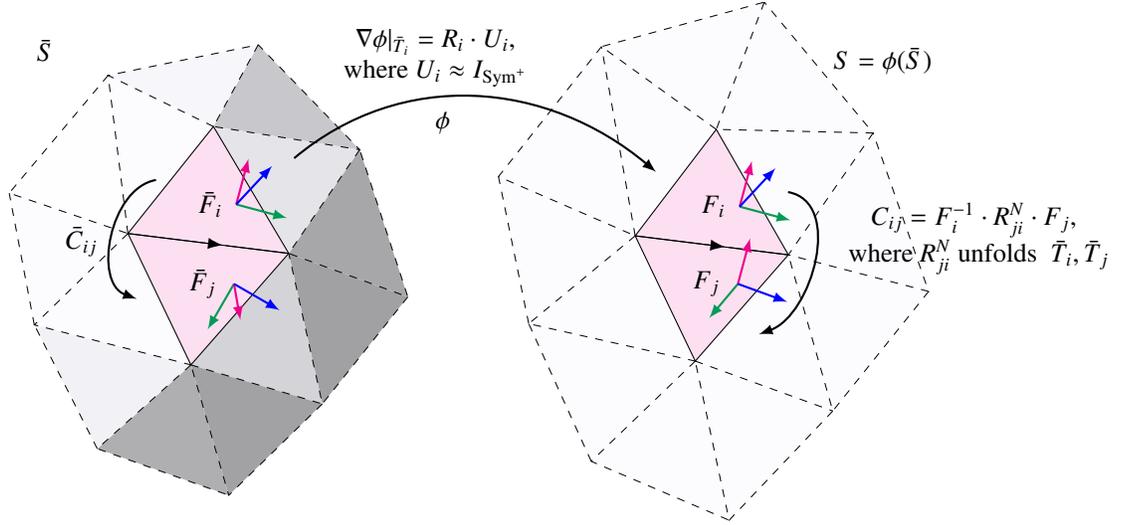
\begin{figure*}[t]
\centering
\begin{tikzpicture}[scale=3.5]

    \coordinate (E1)  at (-0.633747, 0.350197  );
    \coordinate (E2)  at (-0.958541, -0.0564762);
    \coordinate (E3)  at (-1.31385 , -0.403747 );
    \coordinate (E4)  at (-0.81663 , 0.909726  );
    \coordinate (E5)  at (-1.08605 , 0.515523  );
    \coordinate (E6)  at (-1.45757 , 0.0918871 );
    \coordinate (E7)  at (-1.81323 , -0.226681 );
    \coordinate (E8)  at (-1.20192 , 1.30849   );
    \coordinate (E9)  at (-1.37214 , 0.997749  );
    \coordinate (E10) at (-1.69693 , 0.591076  );
    \coordinate (E11) at (-2.05224 , 0.243806  );
    \coordinate (E12) at (-1.77204 , 1.19207   );
    \coordinate (E13) at (-2.14758 , 0.744714  );

    \fill[backtriangleG] (E4)  -- (E5)  -- (E1);
    \fill[backtriangleI] (E2)  -- (E1)  -- (E5);
    \fill[backtriangleE] (E5)  -- (E6)  -- (E2);
    \fill[backtriangleH] (E3)  -- (E2)  -- (E6);
    \fill[backtriangleF] (E6)  -- (E7)  -- (E3);
    \fill[backtriangleD] (E8)  -- (E9)  -- (E4);
    \fill[backtriangleC] (E5)  -- (E4)  -- (E9);
    \fill[fill=fronttriangle] (E9)  -- (E10) -- (E5);  %
    \fill[fill=fronttriangle] (E6)  -- (E5)  -- (E10); %
    \fill[backtriangle] (E10) -- (E11) -- (E6);
    \fill[backtriangleA] (E7)  -- (E6)  -- (E11);
    \fill[backtriangleB] (E9)  -- (E8)  -- (E12);
    \fill[backtriangle] (E10) -- (E9)  -- (E12);
    \fill[backtriangleJ] (E12) -- (E13) -- (E10);
    \fill[backtriangle] (E11) -- (E10) -- (E13);
    
    \draw[dashed, very thin]  (E4)  -- (E5)  -- (E1) -- cycle;
    \draw[dashed, very thin]  (E1)  -- (E2)  -- (E5); 
    \draw[dashed, very thin]  (E6)  -- (E2); 
    \draw[dashed, very thin]  (E2)  -- (E3)  -- (E6); 
    \draw[dashed, very thin]  (E6)  -- (E7)  -- (E3); 
    \draw[dashed, very thin]  (E8)  -- (E9)  -- (E4) -- cycle; 
    \draw (E10)  -- (E9) -- (E5);            %
    \draw (E6)  -- (E5)  -- (E10) -- cycle;  %
    \draw[dashed, very thin]  (E10) -- (E11) -- (E6);
    \draw[dashed, very thin]  (E7)  -- (E11);
    \draw[dashed, very thin]  (E9)  -- (E12)  -- (E8);
    \draw[dashed, very thin]  (E10) -- (E12);
    \draw[dashed, very thin]  (E12) -- (E13) -- (E10);
    \draw[dashed, very thin]  (E11) -- (E13);

    \coordinate (EE1)  at (1.344     ,  0.370973 );
    \coordinate (EE2)  at (0.980065  , -0.0834068);
    \coordinate (EE3)  at (0.586427  , -0.498644 );
    \coordinate (EE4)  at (1.13314   ,  0.97305  );
    \coordinate (EE5)  at (0.811349  ,  0.510023 );
    \coordinate (EE6)  at (0.45777   ,  0.106216 );
    \coordinate (EE7)  at (0.0627517 , -0.274782 );
    \coordinate (EE8)  at (0.719889  ,  1.47034  );
    \coordinate (EE9)  at (0.537551  ,  0.986784 );
    \coordinate (EE10) at (0.22963   ,  0.582247 );
    \coordinate (EE11) at (-0.174685 ,  0.248265 );
    \coordinate (EE12) at (0.101714  ,  1.30145  );
    \coordinate (EE13) at (-0.2896   ,  0.807487 );

    \fill[backtriangle] (EE4)  -- (EE5)  -- (EE1);
    \fill[backtriangle] (EE2)  -- (EE1)  -- (EE5);
    \fill[backtriangle] (EE5)  -- (EE6)  -- (EE2);
    \fill[backtriangle] (EE3)  -- (EE2)  -- (EE6);
    \fill[backtriangle] (EE6)  -- (EE7)  -- (EE3);
    \fill[backtriangle] (EE8)  -- (EE9)  -- (EE4);
    \fill[backtriangle] (EE5)  -- (EE4)  -- (EE9);
    \fill[fill=fronttriangle] (EE9)  -- (EE10) -- (EE5);  %
    \fill[fill=fronttriangle] (EE6)  -- (EE5)  -- (EE10); %
    \fill[backtriangle] (EE10) -- (EE11) -- (EE6);
    \fill[backtriangle] (EE7)  -- (EE6)  -- (EE11);
    \fill[backtriangle] (EE9)  -- (EE8)  -- (EE12);
    \fill[backtriangle] (EE10) -- (EE9)  -- (EE12);
    \fill[backtriangle] (EE12) -- (EE13) -- (EE10);
    \fill[backtriangle] (EE11) -- (EE10) -- (EE13);

    \draw[dashed, very thin]  (EE4)  -- (EE5)  -- (EE1) -- cycle;
    \draw[dashed, very thin]  (EE1)  -- (EE2)  -- (EE5); 
    \draw[dashed, very thin]  (EE6)  -- (EE2); 
    \draw[dashed, very thin]  (EE2)  -- (EE3)  -- (EE6); 
    \draw[dashed, very thin]  (EE6)  -- (EE7)  -- (EE3); 
    \draw[dashed, very thin]  (EE8)  -- (EE9)  -- (EE4) -- cycle; 
    \draw (EE10)  -- (EE9) -- (EE5);            %
    \draw (EE6)  -- (EE5)  -- (EE10) -- cycle;  %
    \draw[dashed, very thin]  (EE10) -- (EE11) -- (EE6);
    \draw[dashed, very thin]  (EE7)  -- (EE11);
    \draw[dashed, very thin]  (EE9)  -- (EE12)  -- (EE8);
    \draw[dashed, very thin]  (EE10) -- (EE12);
    \draw[dashed, very thin]  (EE12) -- (EE13) -- (EE10);
    \draw[dashed, very thin]  (EE11) -- (EE13);

    \coordinate (tiBaryRef) at (barycentric cs:E9=.333,E10=.333,E5=.333);
    \coordinate (tjBaryRef) at (barycentric cs:E6=.333,E5=.333,E10=.333);
    \coordinate (tiBary) at (barycentric cs:EE9=.333,EE10=.333,EE5=.333);
    \coordinate (tjBary) at (barycentric cs:EE6=.333,EE5=.333,EE10=.333);
    \node at (tiBaryRef) {$\reference{F}_i$};
    \node at (tjBaryRef) {$\reference{F}_j$};
    \node at (tiBary) {$F_i$};
    \node at (tjBary) {$F_j$};
    
     \node (e) at ($(E10)!0.65!(E5)$) {};
    \node (e1) at ($(E10)!0.55!(E5)$) {};
    \draw[-latexnew, arrowhead=0.2cm] (E10) -- (e);
    \draw (e1) -- (E5);
    
         \node (ee) at ($(EE10)!0.65!(EE5)$) {};
    \node (ee1) at ($(EE10)!0.55!(EE5)$) {};
    \draw[-latexnew, arrowhead=0.2cm] (EE10) -- (ee);
    \draw (ee1) -- (EE5);

    \node(tiBaryRefFrame) at ([xshift=.1cm]tiBaryRef) {};
    \node(tjBaryRefFrame) at ([xshift=.12cm, yshift=.0cm]tjBaryRef) {};
    
    \node(tiBaryFrame) at ([xshift=.1cm]tiBary) {};
    \node(tjBaryFrame) at ([xshift=.12cm, yshift=.0cm]tjBary) {};
    
    \def\lenlen{.2}
    \def\lenlenlen{.14}
    \def\lenlenlenlen{.18}
    \def\lenlenlenlenlen{.17}
    \def\angle{-15}    
    \draw[ForestGreen,-latex, thick] (tiBaryRefFrame.center) -- ++(\angle:\lenlen);
    \draw[Magenta,-latex, thick] (tiBaryRefFrame.center) -- ++(\angle+90:\lenlenlenlen);
    \draw[blue,-latex, thick] (tiBaryRefFrame.center) -- ++(\angle+62:\lenlen);
    
    \def\angle{190}    
    \draw[ForestGreen,-latex, thick] (tjBaryRefFrame.center) -- ++(\angle-180-130:\lenlen);
    \draw[Magenta,-latex, thick] (tjBaryRefFrame.center) -- ++(\angle+90:\lenlenlen);
    \draw[blue,-latex, thick] (tjBaryRefFrame.center) -- ++(\angle+70-180-110:\lenlen);
    
    \def\angle{-15}    
    \draw[ForestGreen,-latex, thick] (tiBaryFrame.center) -- ++(\angle:\lenlen);
    \draw[Magenta,-latex, thick] (tiBaryFrame.center) -- ++(\angle+90:\lenlenlenlen);
    \draw[blue,-latex, thick] (tiBaryFrame.center) -- ++(\angle+62:\lenlen);
    
    \def\angle{190}    
    \draw[ForestGreen,-latex, thick] (tjBaryFrame.center) -- ++(\angle-180-140:\lenlenlenlenlen);
    \draw[Magenta,-latex, thick] (tjBaryFrame.center) -- ++(\angle-115:\lenlenlenlen);
    \draw[blue,-latex, thick] (tjBaryFrame.center) -- ++(\angle+70-180-100:\lenlen);
    
    \node(halfA) at ($(E10)!0.5!(E9)$) {};
    \node(halfB) at ($(E10)!0.5!(E6)$) {};
    \node(halfAA) at ($(EE5)!0.5!(EE9)$) {};
    \node(halfBB) at ($(EE5)!0.5!(EE6)$) {};
     \node(halfC) at ($(E5)!0.75!(E9)$) {};
     \node(halfCC) at ($(EE10)!0.6!(EE9)$) {};
    \draw[thick, -latex]([xshift=-.02cm]halfA.west) to [out=190,in=160] node(helperh)[left]{} ([xshift=-.05cm, yshift=.0012cm]halfB.west);
    \draw[thick, -latex]([xshift=.1cm]halfAA.east) to [out=350,in=20] node(helperhh)[right]{} ([xshift=.1cm, yshift=-.1cm]halfBB.west);
     \draw[thick, -latex]([xshift=.2cm]halfC.east) to [out=40,in=140] node[above,align=center]{\hspace{-3em}$\nabla\phi|_{\reference{T}_i}=R_i\cdot U_i\text{,}$\\\hspace{-3em}$\text{where }U_i\approx{}I_{\text{Sym}^+}$} ([xshift=-.07cm]halfCC.west);

     \path([xshift=.2cm]halfC.east) to [out=40,in=140] node[below]{\makecell[c]{\hspace{-1.5em}$\phi$}} ([xshift=-.07cm]halfCC.west);
     
     \node(helperh2) at ([xshift=-.08cm, yshift=-.02cm]helperh) {$\reference{C}_{ij}$};
     \node(helperh22)[right,align=center] at ([xshift=+.02cm, yshift=+.08cm]helperhh) {$C_{ij}=F^{-1}_{i}\cdot R_{ji}^{{\color{black} N}}\cdot F_j\text{,}$\\\hspace{0.2em}$\text{\,where }R_{ji}^{{\color{black} N}}$\text{ unfolds } $\reference{T}_i,\reference{T}_j$};

\node[above left](nameReference) at ([yshift=+.22cm]$(E13)!0.55!(E12)$) {$\reference{S}$};
\node[right](nameInstance) at ([yshift=+.05cm]$(EE8)!0.55!(EE4)$) {${S=\phi(\reference{S})}$};

\node(helphelp) at ([xshift=-0.8cm]nameReference) {};

\end{tikzpicture}
\caption{\revised{Non-flat surface $\reference{S}$ is employed as reference within the deformation setup. Flat surface $S$ is determined via deformation of $\reference{S}$ by $\phi$, s.t.\ metric distortion, i.e.\ $U_i$, is close to identity and $R_i$ is determined by means of $C_{ij}$ that are normal vector fixing modifications of $\reference{C}_{ij}$}.} \label{fig:flatteningscheme}
\end{figure*}
As our representation space comprises a symmetric positive-definite and a rotational part the algorithm's respective behavior can be assessed separately.
Since $\Sym[2]$ is \revisedII{a}belian and flat (indeed a vector space) the algorithm converges after exactly one step~\cite{pennec2006intrinsic}.
In contrast, $\SO$ features a less trivial structure exhibiting, e.g.\ a non-empty \emph{cut locus}.
However, as long as the data is located within some $\varepsilon$-ball, with $\varepsilon$ smaller than the injectivity radius of the exponential map,
the existence and (local) uniqueness of the mean can be guaranteed~\cite{pennec2020advances} and thus convergence of the algorithm.
Note that this assumption is only violated for transition rotations differing by more than $\pm\pi$, what can be practically ruled out. (cf.~Appx.~\ref{app:TransRotDist}).
As framework for second order statistics we employ (\revisedII{linearized}) Principal Geodesic Analysis~\cite{Fletcher2004PGA} at $\mu$ that is an extension of the common Principal Component Analysis to Riemannian manifolds allowing for covariance analysis.
In particular, \revisedII{we solve}
\begin{align*}
\vartheta_{p}&= \argmax_{\substack{\vartheta\in T_\mu{}G}}{\sum_{i}{ g^{\mu}_{\omega}{\left(\vartheta, \log_\mu{(s_i)}\right)}^2}},\\
&~\text{s.t. $g_\omega^{\mu}(\vartheta_p,\vartheta_l)=\delta_{pl}$, for $1\leq{}l\leq{}p$}
\end{align*}
for the main modes of variation $\vartheta_p$, where $g_{\omega}$ is the metric associated to distance $d_{\omega}$~(Eq.~\ref{eq:distance}). The solution is found algorithmically by eigendecomposition of the Gram matrix $C=(c_{ij})_{ij}$, with $c_{ij}=g^{\mu}_{\omega}{\left(\log_\mu{(s_i)}, \log_\mu{(s_j)}\right)}$ (cf.~\citet{younes2010shapes} E.2.2).
In order to avoid a systematic bias due to the choice of reference shape $\reference{S}$, we require it to agree with the mean of the training data ($\reference{S}=\reference{S}(\mu)$) as proposed in~\citet{joshi2004unbiasedAtlas}. Details on how to determine a shape for given representation parameters are given in Sec.~\ref{sec:numerics}.}

\subsection{\revised{Quasi-Isometric Surface Flattening}}
\label{sec:flattening}
\revised{Apart from shape analysis, the proposed representation provides an effective framework for processing operations. In this section, we derive an approach for the calculation of a quasi-isometric surface chart, i.e.\ a low-distortion immersion of a given surface into the two dimensional Euclidean space. Since flattening techniques provide a way to access problems of three dimensional context in a two dimensional fashion, such an approach facilitates practically relevant applications like visualization and deep learning based assessment of knee cartilage thickness (Fig.~\ref{fig:flatteningexample}). For a broader overview on application examples we refer to \citet{kreiser2018survey}, who published a survey on flattening-based medical visualization techniques.
The key idea behind our flattening approach is to consider the set of flat immersions of the reference shape $\reference{S}$ as a submanifold in shape space.
This submanifold has a particularly convenient characterization in our representation space $G$ allowing for a simple, isometric projection: We fix the metric part $\{U_i= I_{\text{Sym}^+}\}$ as identity (no metric distortion) and choose transition rotations s.t.\ they act as identity on the normals (zero curvature).
In particular, the latter are given by $\{C_{ij}=F^{-1}_{i}\cdot R_{ji}^{\revisedII{N}}\cdot F_j \}$, where $R_{ji}^{\revisedII{N}}$ unfolds triangle $\reference{T}_j$ to the plane of triangle $\reference{T}_i$.
Phrasing it in the group setting this means we project the transition rotations to $\SO[2]$ (embedded in $\SO$) since all feasible flat shape representations necessarily have to be elements of $\SO[2]^{n} \times \Sym[2]^{m}$.
See Fig.~\ref{fig:flatteningscheme} for a schematic overview.
Note that the obtained projection corresponds to a realizable deformation, iff the input shape $\reference{S}$ is isometric to the plane.
In general, a near-isometric flattening can be efficiently computed using our reconstruction (cf. next section).}
\begin{figure*}[t]
    \fboxsep0pt 
    \center       
    \includegraphics[width=1.0\linewidth]{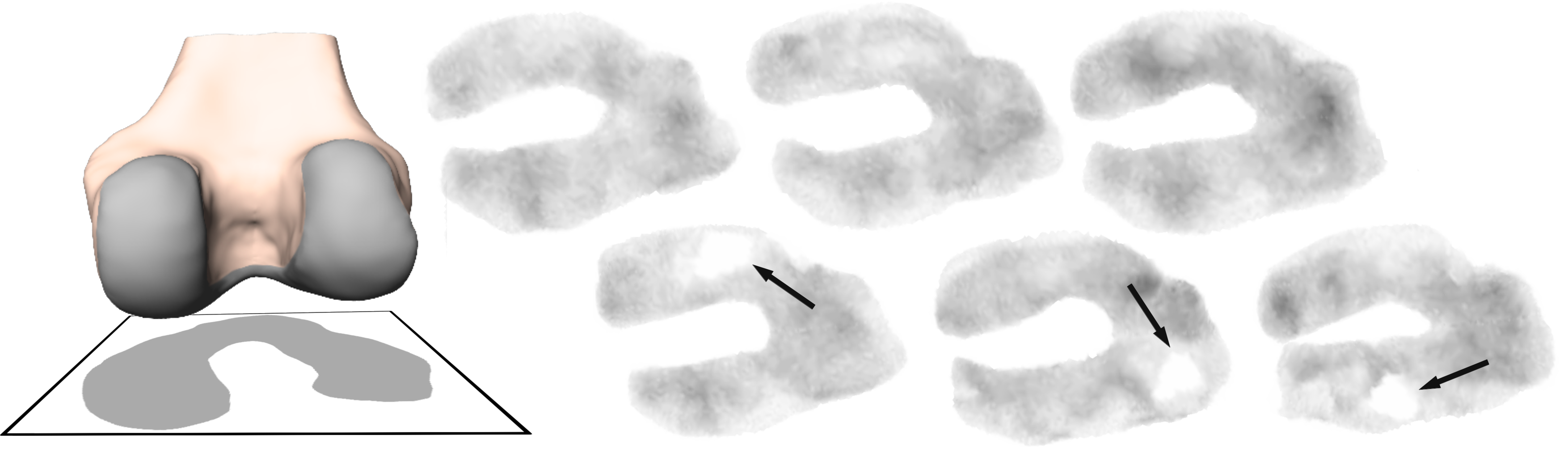}
 \caption{\revised{Left: Example for a flattened femoral articular cartilage region. Right: Flattened femoral cartilage with gray value coded cartilage thickness. The top row shows healthy subjects whereas the subjects in the bottom row exhibit denuded areas within the cartilage region.}}
    \label{fig:flatteningexample} 
\end{figure*}

\section{Efficient Numerics}\label{sec:numerics}
In this section, we propose an efficient numerical algorithm to solve for the inverse problem of mapping a point in representation space $G$ to a corresponding shape $S = \phi(\reference{S})$.
If the corresponding rotations $\{R_i\}$ were known, $\phi$ could be obtained as the minimizer of $\sum_{i=1}^m \reference{A}_i \norm{D_i - R_i \matop U_i}^2_F$ by solving the well-known Poisson equation (see e.g.~\citet{botsch2006deformation,vonTycowicz2018efficient}).
However, in our representation the rotations are only given implicitly in terms of the transition rotations.
In particular, an immediate computation shows that $R_j = R_i \matop \reference{F}_i \matop C_{ij} \matop \reference{F}_j^T =: R_{i \rightarrow j}$ for an integrable field $\{C_{ij}\}$.
Based on this condition, for each triangle $T_i$ we can formulate a residual term $\varepsilon_i(\phi, \{R_i\}) = \sum_{j \in \mathcal{N}_i} \nicefrac{1}{\abs{\mathcal{N}_i}} \norm{D_i - R_{j \rightarrow i} \matop U_i}^2_F$ in terms of the rotations of neighboring triangles (indexed by $\mathcal{N}_i$).
Then, the objective for the inverse problem is given as $E(\phi) = \min_{\{R_i\in\SO\}} E(\phi, \{R_i\})$, where $E(\phi, \{R_i\}) = \sum_{i=1}^m \reference{A}_i \, \varepsilon_i(\phi, \{R_i\})$.
Although $E(\phi)$ is a nonlinear function calling for iterative optimization routines, it exhibits a special structure amenable to an efficient alternating minimization technique.
Specifically, we employ a block coordinate descent strategy that alternates between a local and a global step:

\textit{Local step:} First, we minimize $E(\phi, \{R_i\})$ over the rotations $\{R_i\}$ keeping $\phi$ (hence $D_i$) fixed.
Each summand in $\varepsilon_i$ depends on a single rotation $R_j$, s.t.\ the problem decouples into individual low-dimensional optimizations that can be solved in closed-form and allow for massive parallelization. \revised{Note that the problem reduces to the well-known orthogonal Procrustes problem. Further details are given visually and formally in Appx.~\ref{appx:InitLocalStep}.}

\textit{Global step:} Second, we minimize $E(\phi, \{R_i\})$ over $\phi$ with rotations $\{R_i\}$ fixed leading to a quadratic optimization problem for which the optimality conditions are determined by a Poisson equation.
As the system matrix is sparse and depends only on the reference shape, it can be factorized once during the preprocess allowing for very efficient global solves with close to linear cost.

Note that the objective is bounded from below and that both local and global steps feature unique solutions that are guaranteed to weakly decrease the objective making any numerical safeguards unnecessary.
This contrasts with classical approaches that require precautions, such as line search strategies and modification schemes for singular or indefinite Hessians, to guarantee robustness.

\textbf{Initialization} 
To provide the solver with a warm start, we compute an initial guess for the rotation field $\{R_i\}$.
To this end, we employ the local integrability condition $R_j = R_{i \rightarrow j}$ to propagate an initial rotation matrix from an arbitrary seed along a precomputed spanning tree of the dual graph of $\reference{S}$. Note, that this strategy recovers the rotation field exactly for integrable $\{C_{ij}\}$. In case of non-integrable fields, one advantage of the Poisson-based reconstruction (global step) is that it distributes errors uniformly s.t.\ local inconsistencies are attenuated. \revised{More details on the initialization procedure can be found in Appx.~\ref{appx:InitLocalStep}.}

\section{Experiments and Results}\label{sec:experimentsResults}
\revised{Except where stated otherwise all experiments are performed employing a fixed metric commensuration weight $\omega = 10$ that empirically shows the best performance in our knee osteoarthritis classification experiment (cf. Appx.~\ref{appx:varyCommPara}).}
\subsection{Data}\label{sec:data}
\revised{We employ four different datasets to ensure a qualitative and quantitative as well as a technical and application-oriented assessment of the proposed FCM.\\
{\bf (i) OAI - Right distal femora} (see Fig.~\ref{fig:OAIntro}) from the Osteoarthritis Initiative\footnote{nda.nih.gov/oai} (OAI) database. All subjects are rated w.r.t. knee osteoarthritis using the Kellgren and Lawrence score (0, healthy $\rightarrow$ 4, severely diseased) \cite{kellgren1957KLscore}. The dataset consists of 58 severely diseased (grade 4) and 58 healthy subjects (grade 0,1) that were also used for evaluation in~\cite{vonTycowicz2018efficient} to which we refer for further details on the data, especially with regards to the arrangement of correspondence (8988 vertices, 13776 triangles). We added a list of patient ids to Appx.~\ref{app:DataIds} since the  underlying segmentation masks are publicly available as part of publication~\citet{ambellan2019automated}.\\
{\bf (ii) ADNI - Right hippocampi} (see Fig.~\ref{fig:ADNIIntro}) from the Alzheimer's Disease Neuroimaging Initiative\footnote{adni.loni.usc.edu} (ADNI) consisting of 60 subjects showing Alzheimer's disease and 60 cognitive normal controls. We prepared this dataset using imaging data from the ADNI database that contains, among others, 1632 brain MRI scans collected on four different time points with segmented hippocampi. %
We established surface correspondence (2280 vertices, 4556 triangles) in a fully automatic manner employing the deblurring and denoising of functional maps approach~\cite{ezuz2017deblurring}\footnote{code online available: cs.technion.ac.il/~mirela/code/fmap2p2p.zip} for isosurfaces extracted from the given segmentations.
The dataset was randomly assembled from the baseline shapes for which segmentations were simply connected and remeshed surfaces were well-approximating ($\leq\SI{e-5}{\milli\metre}$ root mean square surface distance to the isosurface). Similar as for the OAI dataset we added a list of scan ids to Appx.~\ref{app:DataIds} since the used hippocampus segmentations are publicly available as part of the ADNI database.\\
{\bf (iii) FAUST - A male human body in two poses} being part of the anthropological, open-access Fine Alignment Using Scan Texture (FAUST)~\citet{Bogo2014FAUST} dataset of whole body scans featuring high-quality, dense correspondences (6890 vertices, 13776 faces).
We chose two poses of the same (male) person, lifting the arms up and down alongside its body (Fig.~\ref{fig:artifactsCurlFAUST}, right).\\
{\bf (iv) PIPE - A pair of synthetic pipe surfaces}, one in a cylindrical and one in a helical configuration consisting of 1220 triangles and 612 vertices as can be seen in Fig.~\ref{fig:artifactsCurlFAUST} (left).\\
 Throughout the manuscript we will refer to the datasets using the above acronyms relating to the data origin or content.
\begin{figure}[tbh]
    \center       
    \includegraphics[width=1\linewidth]{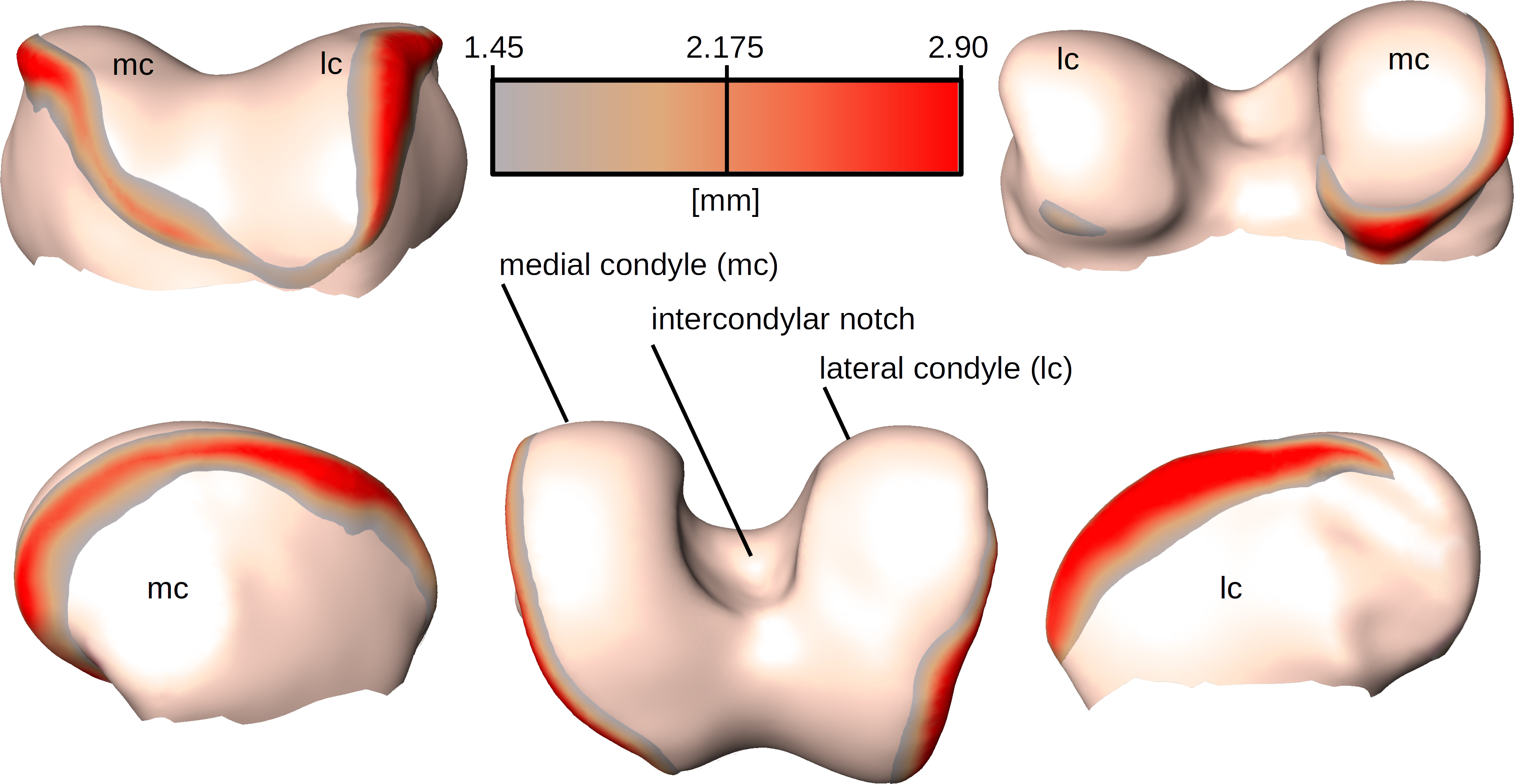}
 \caption{Mean shape of healthy distal femora overlaid with (larger) mean shape of the diseased femora wherever the distance is larger than 1.45mm, colored accordingly.}
    \label{fig:OAIntro} 
\end{figure}
\begin{figure}[tbh]
    \center       
    \includegraphics[width=0.9\linewidth]{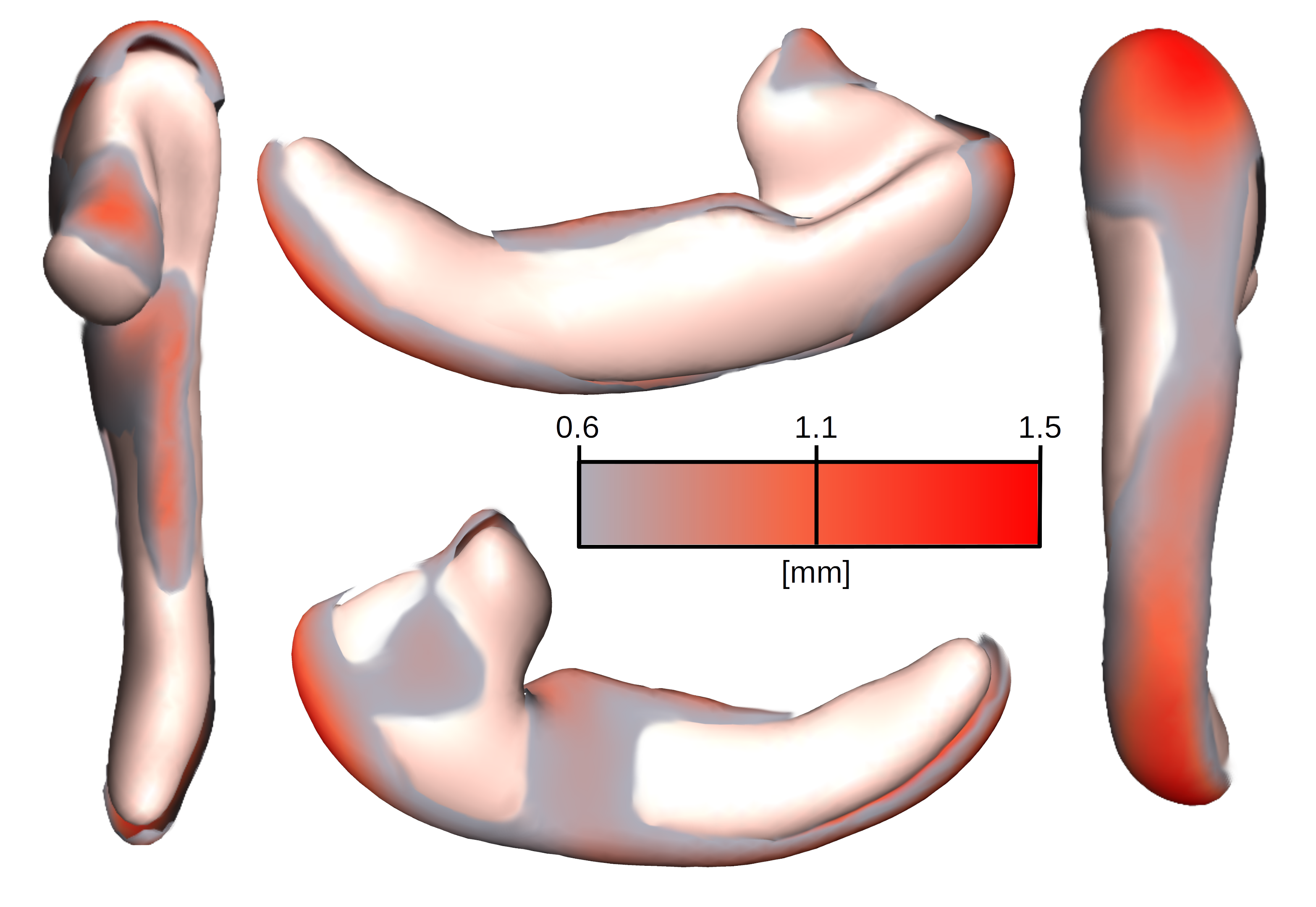}
 \caption{\revised{Mean shape of diseased right hippocampi overlaid with mean shape of the healthy hippocampi wherever the distance is larger than 0.6mm, colored accordingly.}}
    \label{fig:ADNIIntro} 
\end{figure}
\begin{figure}[tbh]
    \center       
    \includegraphics[width=0.99\linewidth]{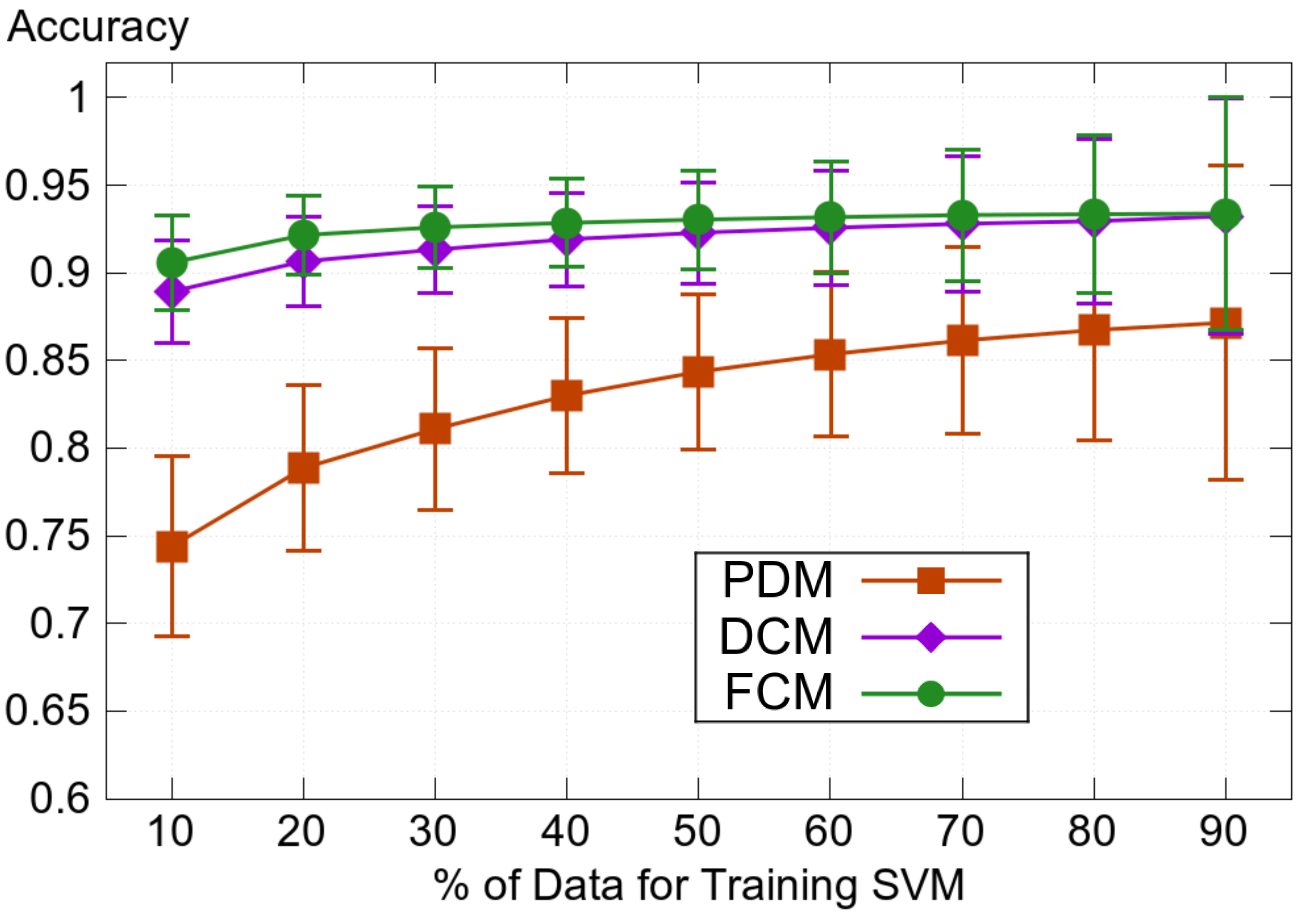}
 \caption{OA classification experiment for the proposed FCM, PDM~\cite{Cootes1995ASM} and DCM~\cite{vonTycowicz2018efficient}.}
    \label{fig:OAClassification} 
\end{figure}
\begin{figure}[tbh]
    \center       
    \includegraphics[width=0.99\linewidth]{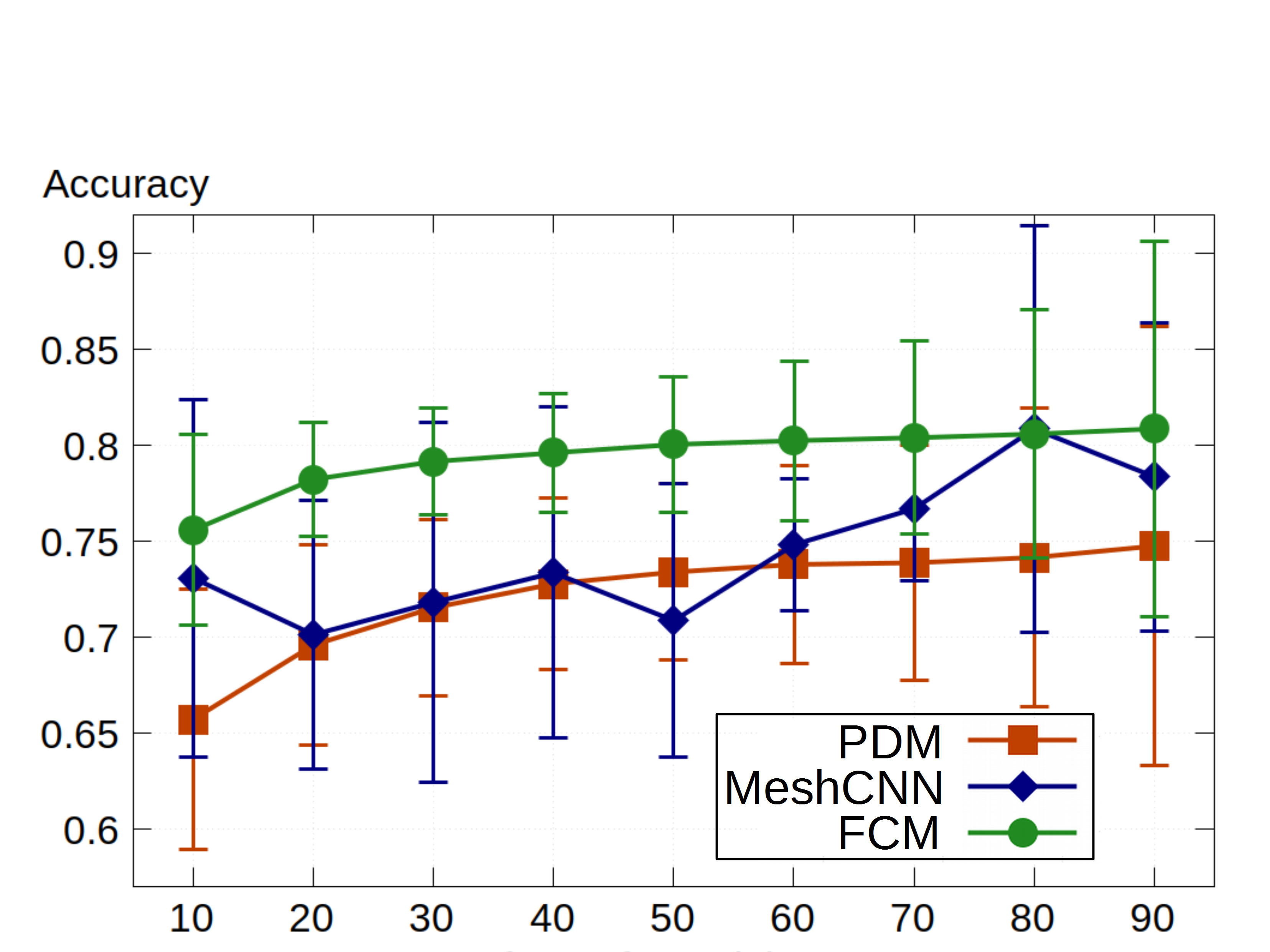}
 \caption{Alzheimer's classification experiment for the proposed FCM\revisedII{, MeshCNN~\cite{hanocka2019meshcnn} and PDM~\cite{Cootes1995ASM}.}}
    \label{fig:ADNIClassification}
\end{figure}
 The datasets OAI and ADNI will be used for quantitative analysis and comparison to other shape models, whereas FAUST and PIPE serve for qualitative assessment of the proposed model.}
\revised{\subsection{Disease Classification}\label{sec:classificationGen}
To assess the sensitivity of the proposed FCM for degenerative shape changes and compare to different other explicit shape modeling approaches, we will perform two binary disease classification experiments, one regarding knee osteoarthritis on the distal femur and one concerning Alzheimer's disease w.r.t.\ the right hippocampus. Although these diseases are very different and not comparable in a direct way we will make use of the same classification pipeline for both of them assessing the generalization potential of the proposed FCM regarding classification tasks. To this end, we train a support vector machine (SVM) with linear kernel on feature vectors comprising shape weights, i.e.\ coefficients of the basis representation in terms of the principal modes for every input shape. By construction this representation is exact for all input shapes and every shape model type. The coefficients hence forthrightly reflect the underlying model approach to gauge variation.
The classifier is trained on a balanced set of feature vectors for different shares of data varying from 10\% to 90\% with testing on the respective complement. To address the randomness in our experimental design, we perform a Monte Carlo cross-validation drawing 10000 times per partitioning.}
\begin{figure*}[t]
\centering
\begin{overpic}[width=0.99\textwidth,percent,grid=false]{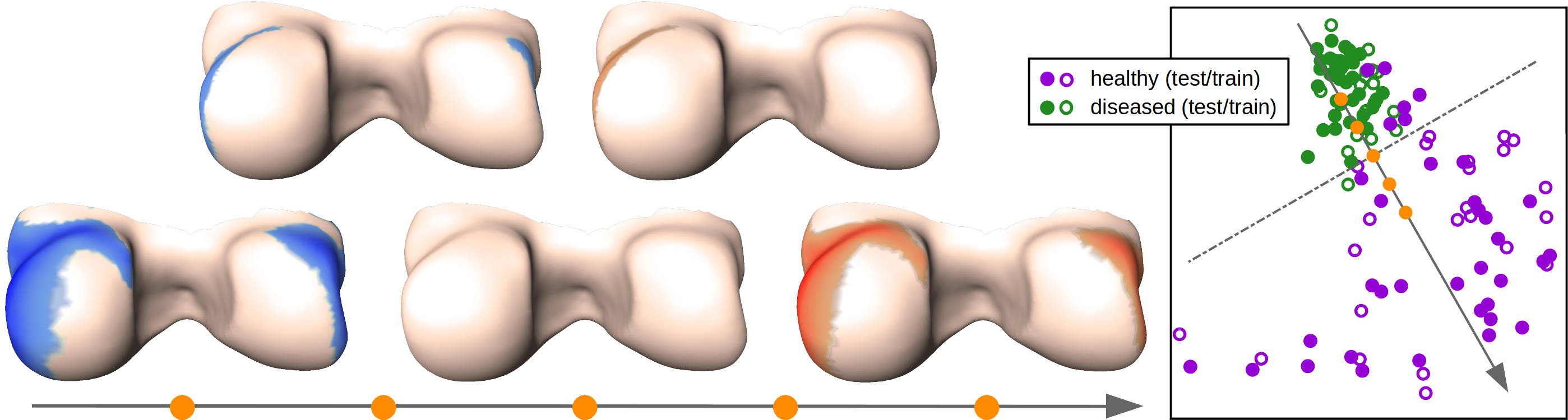}
\put(80,14){$\eta^{\bot}$}
\put(89.5,10){$\eta$}%
\put(43,1.6){$\eta$}%
\end{overpic}
\caption{\revisedII{Visualization of the discriminating direction $\eta$ and separating hyperplane $\eta^\bot$ for OA classification showing a 2-dimensional projection (right) and corresponding shapes (left) equidistantly sampled within the interval containing the input data (note that projections onto $\eta$ and the visualizing plane do not commute causing the interval to appear smaller). Point-wise distance to the middle shape colored using {\it{}-0.5mm}~\ShowColormapOAI{}~{\it{}0.5mm} with neutral window (i.e. rosy color) from {\it{}-0.15mm} to {\it{}0.15mm}.}} \label{fig:svm_oai}  
\end{figure*}

\begin{figure*}[t]
\centering
\centering
\begin{overpic}[width=0.99\textwidth,percent,grid=false]{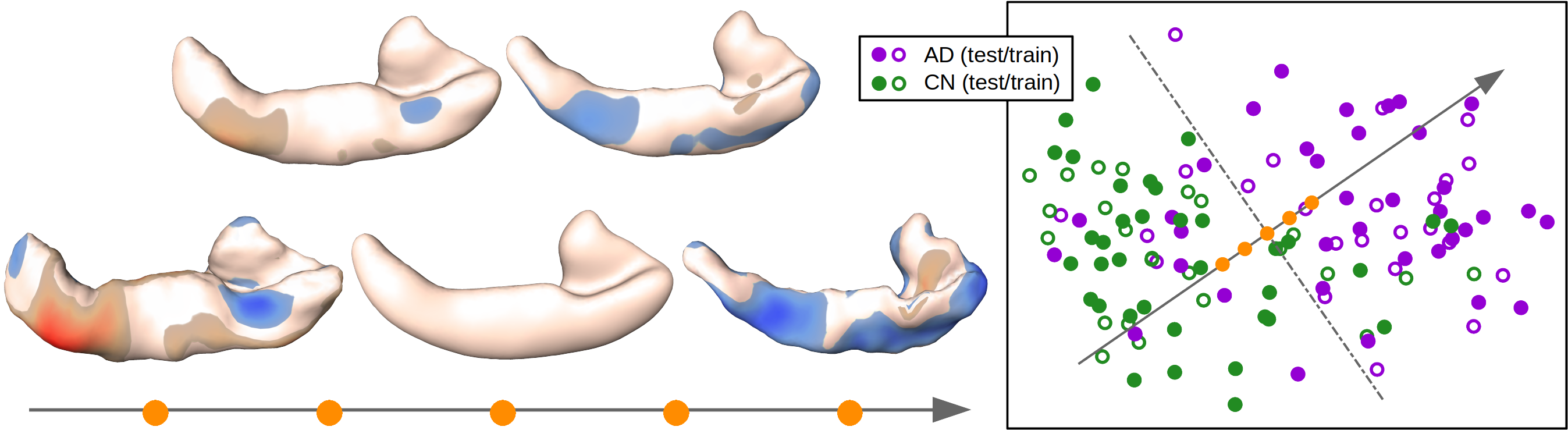}
\put(75.5,21.4){$\eta^{\bot}$}
\put(88,16.15){$\eta$}%
\put(37,2.1){$\eta$}%
\end{overpic}
\caption{\revisedII{Visualization of the discriminating direction $\eta$ and separating hyperplane $\eta^\bot$ for Alzheimer's classification showing a 2-dimensional projection (right) and corresponding shapes (left) equidistantly sampled within the interval containing the input data (note that projections onto $\eta$ and the visualizing plane do not commute causing the interval to appear smaller).  Point-wise distance to the middle shape colored using {\it{}-2.5mm}~\ShowColormap{}~{\it{}2.5mm} with neutral window (i.e. rosy color) from {\it{}-0.5mm} to {\it{}0.5mm}.}} \label{fig:svm_adni}
\end{figure*}
\subsubsection{Knee Osteoarthritis Classification}\label{sec:oaClassification}
Osteoarthritis (OA) is a degenerative disease of the joints that is i.a.\ characterized by changes of the bone shape (see Fig.~\ref{fig:OAIntro}).
Here, we investigate the proposed FCM's ability to classify knee OA for the OAI dataset of distal femora. \revised{Since our test set contains 58 healthy and 58 diseased cases the SVM classifier is trained on 115-dimensional feature vectors.}
We compare to the popular \textit{point distribution model}~\cite{Cootes1995ASM} (PDM) as well as to the differential coordinates model~\cite{vonTycowicz2018efficient} (DCM), which recently achieved highly accurate classification results.
Figure~\ref{fig:OAClassification} shows the results in terms of average accuracy and standard deviation.
Note that solely the FCM achieves an accuracy of over 90\% in case of sparse (10\%) training data.

\begin{figure*}[tb]
 \fboxsep0pt 
  \subfloat{
	\begin{minipage}[c]{0.49\linewidth}
	   \centering
	   \includegraphics[width=1\textwidth]{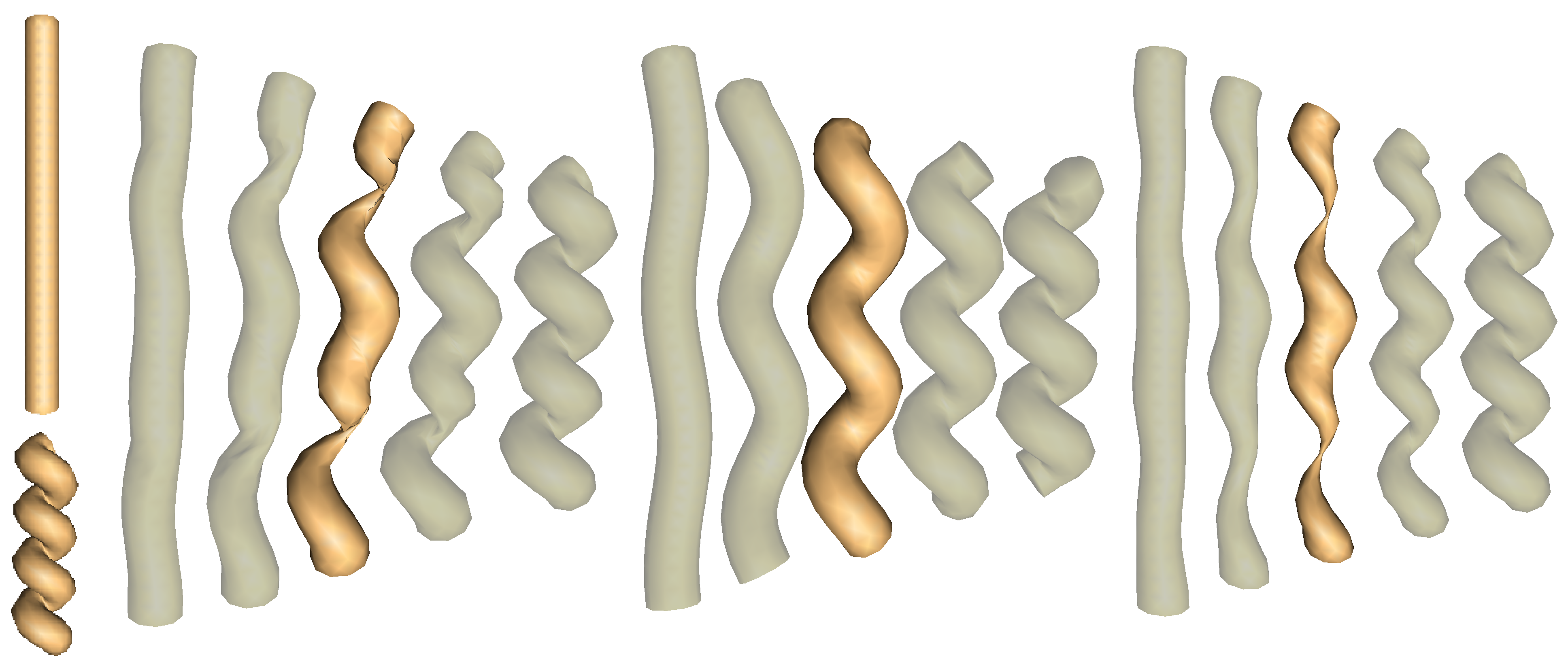}
	\end{minipage}}
 \hfill 	
  \subfloat{
	\begin{minipage}[c]{0.49\linewidth}
	   \centering
	   \includegraphics[width=1\textwidth]{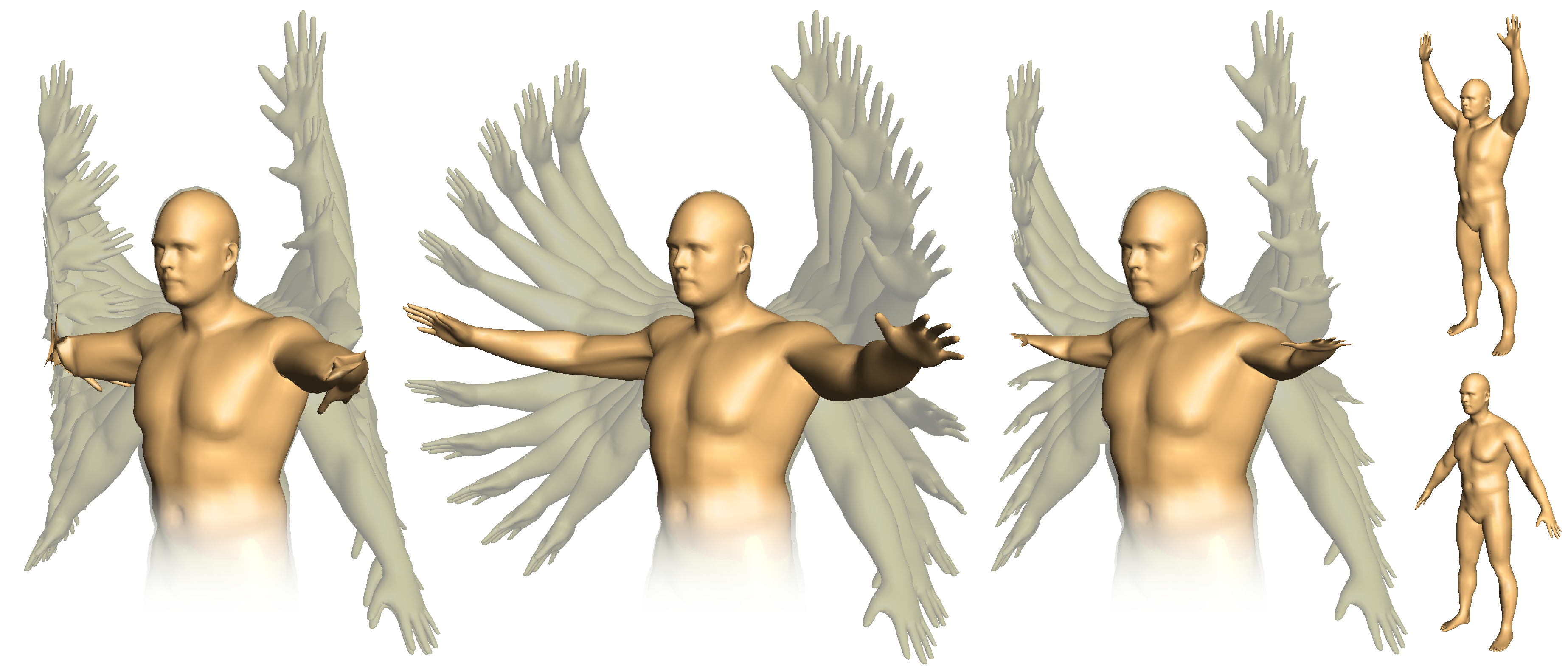}
	\end{minipage}}
\caption{
 Interpolating geodesic (mean highlighted) for the pipe surface (left) and FAUST (right) data within (f.l.t.r.) the DCM~\cite{vonTycowicz2018efficient}, the proposed FCM, and PDM~\cite{Cootes1995ASM}, each.}
 \label{fig:artifactsCurlFAUST} 
\end{figure*}

\revised{\subsubsection{Alzheimer's Classification}\label{sec:AlzClassification}
There is a substantial body of work confirming the well-known connection between hippocampal volume loss and Alzheimer's progression~\cite{kohler1998memory, de2000hemispheric, bonner2015verbal}.
In line with these findings, we observe $\approx\nicefrac{1}{4}$ volume loss between the FCM-based mean shapes of the diseased subjects to the one of the healthy controls.
This motivates a classification experiment regarding the shape of right hippocampi and the disease state to further evaluate the descriptiveness of our shape representation.
For this experiment we employ the commensuration parameter $\omega=0.98$ that empirically performs best w.r.t.\ classification accuracy, i.e.\ metric and curvature related differences are weighted almost equally within the shape analysis.
Since our test set contains 60 cognitive normal and 60 diagnosed Alzheimer's cases the SVM is trained on 119-dimensional feature vectors.

Given the coarse discretization of the hippocampal surface (other than the OAI data) and, thus, moderate hardware requirements, we can perform a direct comparison to MeshCNN~\cite{hanocka2019meshcnn}, i.e.\ a state-of-the-art surface-based classifier from the field of geometric deep learning.
Specifically, we employed the implementation of the authors\footnote{github.com/ranahanocka/MeshCNN} performing training on an Nvidia GTX 980 TI graphics card (6GB memory).
Due to the lack of proper stopping criteria (no option for a validation set), we report the best test accuracy attained in the first 100 epochs, which is rather an upper bound for the classification performance.
Due to the high computational cost ($>2$ hours for training) we restrict to 10 samples per partitioning during Monte Carlo cross-validation, which increases variability as evident by the lack of monotonicity of the estimated dependency of the accuracy on the training size. \revisedIII{Partitioning is carried out analogously to the SVM classifier and training employs the Adam optimizer \cite{kingma2014adam} with weight decay $\beta_1=0.9$ and $\beta_2=0.999$}.

Figure~\ref{fig:ADNIClassification} shows the obtained classification accuracies for MeshCNN as well as our FCM-based and (for reference) PDM-based SVM.
Note that the FCM reaches average accuracies ranging from 75.6\% (10\% training) up to 80.8\% (90\% training) with values above 80\% for all data shares $\geq$50\%.
Remarkably, the FCM-based classifier not only outperforms the PDM one but is also superior to MeshCNN especially in presence of sparse training data. \revisedIII{Note, these results have to be understood in the context of data used, namely the shape of \emph{one} single anatomy. Higher classification accuracy is possible if more data is utilized, as e.g.\ 3D MRI scans of the \emph{whole} brain in \citet{seo2016covariant}.}
}

\revisedII{\subsubsection{Transparency}
The proposed classifier exposes a high degree of interpretability and explainability due to the generative character of statistical shape models and the linearity of the employed SVM.
In particular, the discriminating direction underlying the SVM corresponds to a geodesic in representation space that directly encodes the \textit{single} type of morphological variation that determines the classifier's prediction.
We provide a visualization of the discriminating direction for both experiments (OAI and ADNI) in Fig.~\ref{fig:svm_oai} and~\ref{fig:svm_adni} based on SVM instances with average classification accuracy obtained for 40\%/60\% training/testing split. %
In addition to shapes sampled along the discriminating direction, we provide a 2-dimensional visualization using orthogonal projection onto the plane spanned by the two principal geodesic modes that retain the highest classification accuracy, viz.\ $\vartheta_1,\vartheta_3$ and $\vartheta_1,\vartheta_2$ for ADNI and OAI, respectively.}
\subsection{Validity}\label{sec:validity}
Frequently, datasets feature a high nonlinear variability that are characterized by large rotational components, which are insufficiently captured by linear models like PDM.
While DCM treats the rotational components explicitly, it requires them to be well-localized, s.t.\ the logarithm is unambiguous.
This assumption may not be satisfied for data with large spread in shape space.
Contrary, our model overcomes this limitation by utilizing a relative encoding via transition rotations, which will never exceed $\pm\pi$ in practical scenarios \revised{(cf. Appx.~\ref{app:TransRotDist})}.
In~Fig.~\ref{fig:artifactsCurlFAUST} we illustrate the validity of our model for two extreme examples in comparison to PDM and DCM.
\subsection{Computational Performance}\label{sec:performance}
We compare our framework in terms of computational efficiency to two state-of-the-art approaches:
The \textit{large deformation diffeomorphism metric mapping} (LDDMM) using the open-source Deformetrica~\cite{Durrleman2014deformetrica} software, and the recent DCM.
To this end, we compute the mean shape on 100 randomly sampled pairs from the OAI dataset.
Overall, the LDDMM approach requires 172.8s ($\pm$44.8s) in average whereas the proposed FCM features an average runtime of only 2.3s ($\pm$1.9s), hence a two orders of magnitude speedup.
In comparison to the highly efficient DCM---requiring 1.1s ($\pm$0.3s) in average---our model achieves runtimes within the same order of magnitude, despite the added nonlinearity in the inverse problem.
\subsection{Specificity, Generalization Ability, Compactness}\label{sec:stdMeasures}
We perform a quantitative comparison with PDM and DCM using standard measures (detailed in \citet{Davies2008SMS}) w.r.t.\ a physically-based surface distance $\mathcal{W}$~\cite{heeren2018ShellPGA} as proposed in  \citet{vonTycowicz2018efficient}.
\textit{Specificity} (Fig.~\ref{fig:SpecificityGenAbilityCompactness} middle) evaluates the validity of the model generated instances in terms of their distance to the training shapes. 
We estimate it using 1000 randomly generated instances according to the discrete distribution of the respective model.
\textit{Generalization ability} (Fig.~\ref{fig:SpecificityGenAbilityCompactness} left) assesses how well a model represents unseen instances. It is calculated in a leave-one-out study.
\textit{Compactness} (Fig.~\ref{fig:SpecificityGenAbilityCompactness} right) measures the relative amount of variability of the training set captured by every mode in an accumulated manner.
The results show that the FCM is more specific than PDM and DCM.
In terms of generalization ability, the FCM is superior to PDM, yet inferior to DCM.
Finally, the FCM is less compact than PDM and DCM. Note that compactness is calculated for each model w.r.t its own metric, hence not directly comparable. In particular, we found that decreasing $\omega$ leads to increased compactness, albeit at the \revised{possible} expense of classification accuracy \revised{(cf. Appx.~\ref{appx:varyCommPara}}).
\begin{figure*}[t]
    \fboxsep0pt 
    \center       
    \includegraphics[width=0.99\linewidth]{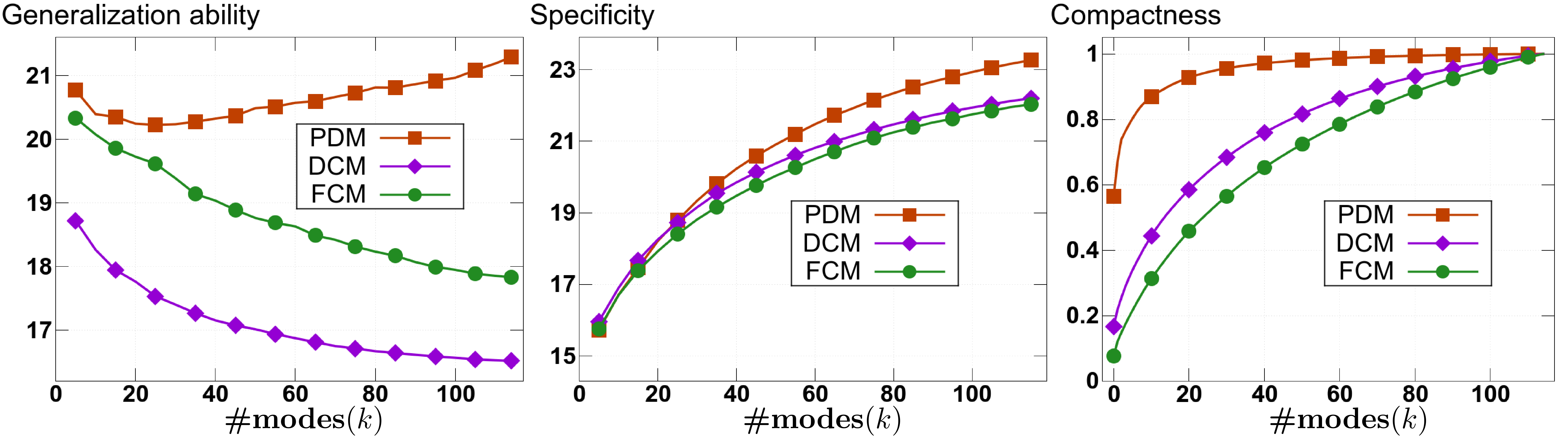}
\caption{Generalization ability \revisedII{(left, lower$\sim$better)}, specificity \revisedII{(middle, lower$\sim$more specific)} and compactness \revisedII{(right, higher$\sim$more compact)} of the proposed FCM, PDM~\cite{Cootes1995ASM}, and DCM~\cite{vonTycowicz2018efficient} on the OAI dataset.}
    \label{fig:SpecificityGenAbilityCompactness} 
\end{figure*}
\section{Conclusion and Future Work}\label{sec:conclusion}
In this work, we presented a novel nonlinear SSM based on a Euclidean motion invariant---hence alignment-free---shape representation with deep foundations in surface theory. The rich structure of the derived shape space assures valid shape instances even in presence of strong nonlinear variability.
\revised{Moreover, we demonstrated that the proposed shape representation can be used to effectively calculate quasi-isometric flat immersions to the plane.}
We perform\revised{ed} manifold-valued statistics in a consistent Lie group setup allowing for closed-form evaluation of Riemannian operations.
Furthermore, we devis\revised{ed} an efficient and robust algorithm to solve the inverse problem that does not require any numerical safeguards.

\revisedII{
We showed that FCM yields highly differentiating shape descriptors that promote state-of-the-art performance for shape-based disease state classification:
(1) In comparisons based on a simple classifier (viz.\ linear SVM) our descriptors are superior to the recent Riemannian DCM~\cite{vonTycowicz2018efficient} as well as the popular linear PDM~\cite{Cootes1995ASM} based descriptors;
(2) Remarkably, the FCM-based SVM significantly outperforms the state-of-the-art, geometric deep learning approach MeshCNN~\citep{hanocka2019meshcnn}.

We would like to remark that our approach guarantees deformations to be only locally diffeomorphic (i.e.\ immersions) but not globally.
However, we did not observe any non-diffeomorphic instances in our experiments (e.g.\ all FCM-derived shapes shown in this article are embeddings).
Indeed, the FCM correctly captures nonlinear deformations with large rotational components that violate well-localizedness assumptions of previous approaches.
On the other hand, in comparison to shape spaces based on diffeomorphic mapping, FCM allows for fast processing of large-scale shape collections and is invariant under Euclidean motion, hence, not susceptible to any bias due to misalignment.
}

One possible and interesting way to proceed in the future is to replace the log-Euclidean metric with the affine-invariant one, which can be considered the natural metric on the symmetric positive-definite matrices.
\revised{Another interesting line of future work is to explore the feasibility of fully automatic computer-aided diagnostics based on advanced machine learning, e.g.\ by combining our shape representation with the approach in~\citet{vonTycowicz2020GCN} or utilizing the proposed flattening approach to transform three-dimensional problems into the well-known two-dimensional image-based deep learning setup.}

\begin{sloppypar}
\subsection*{Acknowledgments}
The authors are funded by the Deutsche Forschungsgemeinschaft (DFG, German Research 
Foundation) under Germany's Excellence Strategy – The Berlin Mathematics 
Research Center MATH+ (EXC-2046/1, project ID: 390685689). This work partly relies on data of the Osteoarthritis Initiative\footnote{The Osteoarthritis Initiative is a public-private partnership comprised of five contracts 
(N01-AR-2-2258; N01-AR-2-2259; N01-AR-2-2260; N01-AR-2-2261; N01-AR-2-2262) 
funded by the National Institutes of Health, a branch of the Department of Health and 
Human Services, and conducted by the OAI Study Investigators. Private funding partners 
include Merck Research Laboratories; Novartis Pharmaceuticals Corporation, 
GlaxoSmithKline; and Pfizer, Inc. Private sector funding for the OAI is managed by the 
Foundation for the National Institutes of Health. This manuscript was prepared using an OAI 
public use dataset and does not necessarily reflect the opinions or views of the OAI 
investigators, the NIH, or the private funding partners.}
\revised{and on data of the Alzheimer's Disease Neuroimaging Initiative (ADNI)\footnote{Data collection and sharing for this project was funded by the ADNI (National Institutes of Health Grant U01 AG024904) and DOD ADNI (Department of Defense award number W81XWH-12-2-0012). ADNI is funded by the National Institute on Aging, the National Institute of Biomedical Imaging and Bioengineering, and through generous contributions from the following: AbbVie, Alzheimer's Association; Alzheimer's Drug Discovery Foundation; Araclon Biotech; BioClinica, Inc.; Biogen; Bristol-Myers Squibb Company; CereSpir, Inc.; Cogstate; Eisai Inc.; Elan Pharmaceuticals, Inc.; Eli Lilly and Company; EuroImmun; F. Hoffmann-La Roche Ltd and its affiliated company Genentech, Inc.; Fujirebio; GE Healthcare; IXICO Ltd.;Janssen Alzheimer Immunotherapy Research \& Development, LLC.; Johnson \& Johnson Pharmaceutical Research \& Development LLC.; Lumosity; Lundbeck; Merck \& Co., Inc.;Meso Scale Diagnostics, LLC.; NeuroRx Research; Neurotrack Technologies; Novartis Pharmaceuticals Corporation; Pfizer Inc.; Piramal Imaging; Servier; Takeda Pharmaceutical Company; and Transition Therapeutics. The Canadian Institutes of Health Research is providing funds to support ADNI clinical sites in Canada. Private sector contributions are facilitated by the Foundation for the National Institutes of Health (www.fnih.org). The grantee organization is the Northern California Institute for Research and Education, and the study is coordinated by the Alzheimer's Therapeutic Research Institute at the University of Southern California. ADNI data are disseminated by the Laboratory for Neuro Imaging at the University of Southern California.}.}
Furthermore we are grateful for the open-access dataset FAUST~\cite{Bogo2014FAUST} as well as for the open-source software Deformetrica~\cite{Durrleman2014deformetrica}. \revised{Finally, we want to thank A.~Tack for providing us with cartilage thickness measurements for Fig.~\ref{fig:flatteningexample}. The present article is an 
extended version of \cite{ambellan2019surface} presented at MICCAI~2019.}

\end{sloppypar}

\appendix
\newcounter{ssubfigure}
\setcounter{subfigure@save}{0}
\renewcommand\thefigure{\thesubsection.\arabic{figure}}
\renewcommand\thesubfigure{\thesubsection.\arabic{figure}.\arabic{subfigure@save}}
\renewcommand\theequation{\thesubsection.\arabic{equation}}
\renewcommand\thetable{\thesubsection.\arabic{table}}

\section*{Appendix}  
\setcounter{figure}{0}
\setcounter{ssubfigure}{0}
\setcounter{subfigure@save}{0}
\renewcommand{\thesubsection}{\Alph{subsection}}
\subsection{Shape Reconstruction - Initialization and Local Step}
\label{appx:InitLocalStep}
\revised{Efficient shape reconstruction as outlined in Sec.~\ref{sec:numerics} is an essential part of this work and consists of an initialization and an iteration of global and local step until a solution is reached. Since the global step is basically solving a Poisson problem and also exemplified in~\citet{vonTycowicz2018efficient} we will omit it here.
To ease deeper inside on how to initialize the reconstruction algorithm and how to do the local step we provide two schematic visuals and explicitly determine the solution of the local step through direct calculation.\\
{\bf Initialization.} To initialize the algorithm we fix the pose of an arbitrarily chosen triangle $i_0$ by fixing its rotation $R_{i_0}$ relative to reference shape $\reference{S}$. Starting from triangle $i_0$ a spanning tree is determined to define a path through the dual graph, passing every triangle exactly once. 
}
 \begin{figure}[tbh]
 \centering
  \subfloat{%
  \begin{minipage}{0.966\linewidth}
      \resizebox{\linewidth}{0.7637\linewidth}{\input{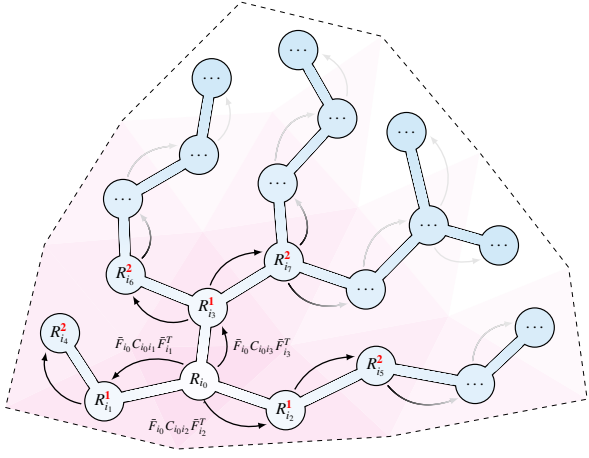}}%
      \end{minipage}%
  }
 \caption{\revised{Initialization procedure propagating an initial rotation $R_{i_0}$ along a pre-computed spanning tree across the data. The red digits indicate the number of propagation steps needed to reach a certain point.}} \label{fig:Initialization}

\end{figure}

\revised{\noindent
$R_{i_0}$ is propagated along this path employing the local integrability condition, viz.\ $R_j := R_{i \rightarrow j}=R_i\reference{F}_i C_{ij} \reference{F}_j^{T}$. Finally this procedure provides a field of extrinsic rotations $\{R_i\}$ to initialize the the local/global solver. Note that the algorithm will stop right after the first iteration, if the $\{C_{ij}\}$ form an integrable system. In that case the initialization is additionally invariant w.r.t. the choice of the fixed triangle and spanning tree. Figure~\ref{fig:Initialization} provides a schematic overview on the initialization process.\\
{\bf Local Step.} Within the local step of the reconstruction algorithm the deformation $\phi$ of $\reference{S}$ and hence the deformation gradient is fixed. We aim to find $R_i$ for every triangle $i$, s.t.\ it, mediated through $\{C_{ij}\}$, locally optimally accompanies $\phi$.  \guillemotright{}Local\guillemotleft{} hereby has to be understood as the one-ring triangle-neighborhood. This problem has a closed-form solution that we will work out within the following proposition via direct calculation on the optimization target.\\
~\\
\newtheorem{theo}{Proposition}
\newtheorem*{theoNo}{Proposition}
\begin{theoNo}
The Local Step within the shape reconstruction algorithm targeting the optimization problem:
\begin{align*}
    R_i = \argmin_{R\in \SO } \sum_{s\in\mathcal{N}_i} \norm{\nabla\phi|_{\reference{T}_s} - R\reference{F}_i C_{is}\reference{F}_{s}^{T} U_{s}}_F^2,
\end{align*}
where $\mathcal{N}_i$ is the set of indices belonging to edge neighbors of triangle $i$,
can be solved in closed form and the solution is unique.
\end{theoNo}
\begin{proof}
For the sake of simplicity let $D_s=\nabla\phi|_{\reference{T}_s}$. We carry out a direct calculation utilizing the definition $\scal{A}{B}_F:= \tr{(A^TB)}$ and the trace's invariance under cyclic permutations:
\begin{align*}
    R_i &= \argmin_{R\in \SO } \sum_{s\in\mathcal{N}_i} \norm{D_{s} - R\reference{F}_i C_{is}\reference{F}_{s}^{T} U_{s}}_F^2\\
    &=\argmin_{R\in \SO } \sum_{s\in\mathcal{N}_i} \underbrace{\norm{D_s}_F^2}_{\text{const.}} - 2\scal{D_s}{R\reference{F}_i C_{is}\reference{F}_s^{T} U_s}_F + \underbrace{\norm{R\reference{F}_i C_{is}\reference{F}^{T}_s U_s}_F^2}_{\text{const.}} \\
    &=\argmax_{R\in \SO } \sum_{s\in\mathcal{N}_i}\tr{\left(D_s^T{R\reference{F}_i C_{is}\reference{F}_s^{T} U_s}\right)}\\
    &=\argmax_{R\in \SO }\sum_{s\in\mathcal{N}_i}\scal{D_sU_s^T\reference{F}_s C_{is}^T\reference{F}_i^{T}}{R}_F=\argmax_{R\in \SO }\big\langle D_{\mathcal{N}_i}, {R}\big\rangle_F
\end{align*}
Since $D_{\mathcal{N}_i}$ is a nonsingular and orientation-preserving matrix it can be uniquely decomposed via polar decomposition to ${R}_{\mathcal{N}_i}{U}_{\mathcal{N}_i}$, where ${R}_{\mathcal{N}_i}\in\SO$ and ${U}_{\mathcal{N}_i}\in \Sym$ s.t. 
\begin{align*}
{R}_{\mathcal{N}_i}=\argmax_{R\in \SO }\big\langle D_{\mathcal{N}_i}, {R}\big\rangle_F.
\end{align*}
\end{proof}

Figure~\ref{fig:LocalStep} schematically illustrates the underlying neighboring relations framing the local integrability constraints.
}

\begin{figure}[tbh]
\centering
  \subfloat{%
  \begin{minipage}{0.966\linewidth}
      \resizebox{\linewidth}{0.7637\linewidth}{\begin{tikzpicture}[scale= 0.8, rotate=0.0]
\coordinate (E1) at  (-2.87331 , -3.69042 );
\coordinate (E2) at  (-6.52804 , -2.81035 );   
\coordinate (E3) at  ( 1.64524 , -3.42601 );
\coordinate (E4) at  (-5.69997 ,  0.638607); 
\coordinate (E6) at  ( 5.4649  , 0.158519); 
\coordinate (E7) at  (-4.04656 ,  3.26233 );  
\coordinate (E8) at  ( 1.48398 ,  5.05326 );   
\coordinate (E9) at  ( 3.57431 ,  2.58582 );
\coordinate (E10) at (-1.51585 ,  5.77764 );
\coordinate (E11) at (-3.87319 ,  -1.53949);
\coordinate (E12) at (-0.451449,  -1.48617);
\coordinate (E13) at (-2.32946 ,  0.914473);
\coordinate (E14) at ( 2.91061 ,  -0.87626);
\coordinate (E15) at ( 0.982103,  1.4311  );
\coordinate (E16) at (-0.845718,  3.45062 );
\coordinate (E17) at  ( 5.86202 , -2.6336  );

\fill[backtriangle] (E12) -- (E1)  -- (E11) -- cycle;
\fill[backtriangle] (E11) -- (E1)  -- (E2)  -- cycle;
\fill[backtriangle] (E2)  -- (E4)  -- (E11) -- cycle;
\fill[backtriangle] (E12) -- (E3)  -- (E1)  -- cycle;
\fill[backtriangle] (E11) -- (E4)  -- (E13) -- cycle;
\fill[backtriangle] (E14) -- (E3)  -- (E12) -- cycle;
\fill[backtriangle] (E16) -- (E13) -- (E7)  -- cycle;
\fill[backtriangle] (E3)  -- (E14) -- (E17) -- cycle;
\fill[backtriangle] (E6)  -- (E14) -- (E17) -- cycle;
\fill[backtriangle] (E7)  -- (E13) -- (E4)  -- cycle;
\fill[backtriangle] (E9)  -- (E14) -- (E15) -- cycle;
\fill[backtriangle] (E6)  -- (E14) -- (E9)  -- cycle;
\fill[backtriangle] (E8)  -- (E16) -- (E10) -- cycle;
\fill[backtriangle] (E15) -- (E16) -- (E8)  -- cycle;
\fill[backtriangle] (E10) -- (E16) -- (E7)  -- cycle;
\fill[backtriangle] (E9)  -- (E15) -- (E8)  -- cycle;

\draw[dashed, very thin] (E12) -- (E1)  -- (E11) ;
\draw[dashed, very thin] (E11) -- (E2)  -- (E1)  ;
\draw[dashed, very thin] (E2)  -- (E4) ;
\draw[dashed, very thin] (E3)  -- (E1) ;
\draw[dashed, very thin] (E11) -- (E4) ;
\draw[dashed, very thin] (E3)  -- (E12);
\draw[dashed, very thin] (E14)  -- (E3) -- (E17) ;
\draw[dashed, very thin] (E6)  -- (E17) -- (E14) ;
\draw[dashed, very thin] (E7)  -- (E13) -- (E4) -- (E7) ;
\draw[dashed, very thin] (E6)  -- (E14) -- (E9) -- (E6) ;
\draw[dashed, very thin] (E16) -- (E8) -- (E10)  ;
\draw[dashed, very thin] (E8)  -- (E9);
\draw[dashed, very thin] (E10) -- (E16) -- (E7) -- (E10) ;
\draw[dashed, very thin] (E9)  -- (E15) -- (E8)  ;

\draw[fill=fronttriangle] (E12) -- (E13) -- (E15) -- cycle; %

\fill[c6] (E12) -- (E11) -- (E13); %
\fill[c6] (E12) -- (E15) -- (E14); %
\fill[c6] (E15) -- (E13) -- (E16); %

\draw (E12) -- (E11) -- (E13); %
\draw (E12) -- (E14) -- (E15); %
\draw (E13) -- (E16) -- (E15); %

\coordinate (tiBaryRef) at (barycentric cs:E12=.333,E13=.333,E15=.333);
\coordinate (tjBaryRef) at (barycentric cs:E12=.15,E11=.7,E13=.15);
\coordinate (tkBaryRef) at (barycentric cs:E12=.15,E15=.15,E14=.7);
\coordinate (tlBaryRef) at (barycentric cs:E15=.15,E13=.15,E16=.7);

\node at (tiBaryRef) {$R_i$};
\node at ([xshift=.1cm]tjBaryRef) {$R_{i\rightarrow j}$};
\node at (tkBaryRef) {$R_{i\rightarrow k}$};
\node at ([xshift=.07cm]tlBaryRef) {$R_{i\rightarrow l}$};

\node (Rk) at ($(tkBaryRef)!0.2!(tiBaryRef)$) {};
\node (RkRi) at ($(tkBaryRef)!0.78!(tiBaryRef)$) {};
\node (RlRi) at ($(tlBaryRef)!0.88!(tiBaryRef)$) {};
\node (RjRi) at ($(tjBaryRef)!0.88!(tiBaryRef)$) {};

\draw[latexnew-, arrowhead=0.15cm, dotted, thick]([xshift=-.1cm, yshift=.28cm]tkBaryRef.center) to [out=95,in=20] node[pos=.55, below] (Mk) {} ([xshift=-.5cm, yshift=.3cm]RkRi.east);
\node (MMk) at ([xshift=.08cm, yshift=-.4cm]Mk) {\begin{scriptsize}$\reference{F}_i C_{ik} \reference{F}_k^{T}$\end{scriptsize}};

\draw[latexnew-, arrowhead=0.15cm, dotted, thick]([xshift=-.25cm, yshift=-.2cm]tlBaryRef.west) to [out=210,in=160] node[pos=.55] (Ml) {} ([xshift=-.09cm, yshift=-.2cm]RlRi.west);
\node (MMl) at ([xshift=.8cm, yshift=.4cm]Ml) {\begin{scriptsize}$\reference{F}_i C_{il} \reference{F}_l^{T}$\end{scriptsize}};

\draw[latexnew-, arrowhead=0.15cm, dotted, thick]([xshift=.65cm, yshift=-.05cm]tjBaryRef.east) to [out=-10,in=270] node[pos=.55] (Mj) {} ([xshift=.25cm, yshift=.1cm]RjRi.south);
\node (MMj) at ([xshift=-.9cm, yshift=.15cm]Mj) {\begin{scriptsize}$\reference{F}_i C_{ij} \reference{F}_j^{T}$\end{scriptsize}};

\node (e) at ($(E12)!0.85!(E13)$) {};
\node (e1) at ($(E12)!0.65!(E13)$) {};
\draw[-latexnew, arrowhead=0.2cm] (E12) -- (e);
\draw (e1) -- (E13);

\node (e) at ($(E15)!0.8!(E12)$) {};
\node (e1) at ($(E15)!0.65!(E12)$) {};
\draw[-latexnew, arrowhead=0.2cm] (E15) -- (e);
\draw (e1) -- (E12);

\node (e) at ($(E13)!0.8!(E15)$) {};
\node (e1) at ($(E13)!0.65!(E15)$) {};
\draw[-latexnew, arrowhead=0.2cm] (E13) -- (e);
\draw (e1) -- (E15);

\end{tikzpicture}}%
      \end{minipage}%
  }
 \caption{\revised{Neighboring relations employed in the local step of the integration procedure. Rotations $R_i$ are connected by transition rotations $C_{ij}$} \label{fig:LocalStep}}
 \end{figure}
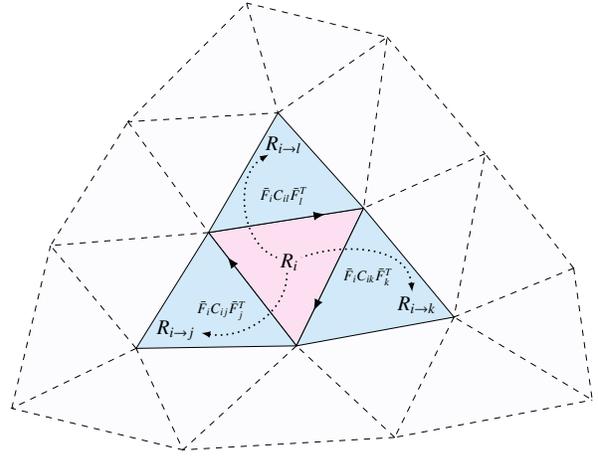

\subsection{Rotational Logarithm and Relative Transition Rotations}
\setcounter{figure}{0}
\label{app:TransRotDist}
\begin{figure*}[tbh]
\centering
\input{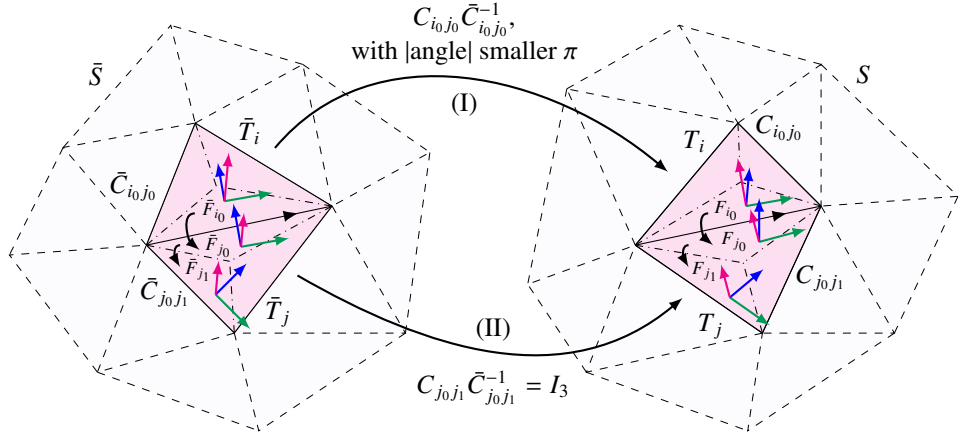}
\caption{\revised{Schematic summary of the argumentation to proof the proposition on relative transition rotations. The given construction allows to separate relative transition rotations into normal (I) and tangential (II) type.}} \label{fig:angleProp}
\end{figure*}
\revised{
As with other non-Euclidean approaches, existence and uniqueness of the intrinsic mean is only ensured for well-localized data.
In particular, for our representation this concerns the rotational components describing the changes in curvature. \revised{We would like to remark that this is a rather academic discussion as we did not encounter any example with critical disparity.
Indeed, even for the synthetic PIPE dataset representing a severe nonlinear deformation the relative transition rotations are located in a small neighborhood of radius $\nicefrac{5\pi}{23}$, see Fig.~\ref{fig:histogramm}.  
However, the following proposition explains how to (theoretically) control the {\it relative transition rotations} and thus how to avoid ambiguities regarding the rotational logarithm.

\begin{figure}[tbh]
    \fboxsep0pt 
    \center       
    \includegraphics[width=1.0\linewidth]{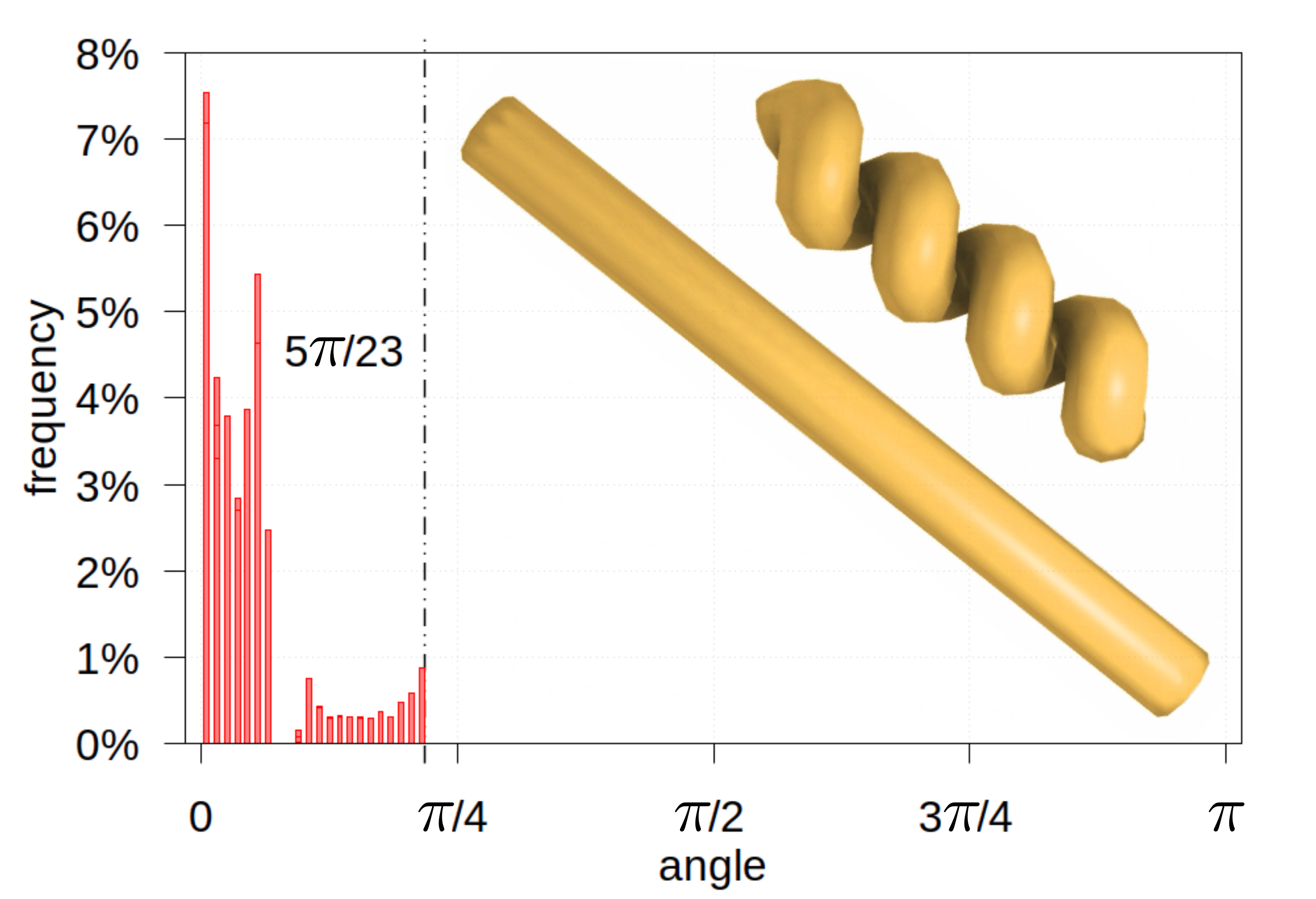}
 \caption{\revised{Histogram of angles between transition rotations of the PIPE shapes. The angles are all relatively small and far away from the critical region of angles larger than $\pm\pi$}.}
 \label{fig:histogramm} 
\end{figure}
}
\begin{theoNo}
For any given $n$ shapes ${S_1, \ldots, S_n}$ there exists a common discetization and a frame field, such that all relative transition rotations exhibit angles in $(-\pi, \pi)$.
\end{theoNo}
}
\revised{
\begin{proof}
We poof the result for two shapes $\reference{S}$ and $S$ since the argumentation naturally extends to the whole $n$ shapes.
At first we have to control change in normal direction in regions of high curvature. We therefore simultaneously refine the triangulations until the angle between normals of any two neighboring triangles lies in $\left(-\nicefrac{\pi}{2}, \nicefrac{\pi}{2}\right)$. Form this point onwards the argumentation is additionally summarized in Fig.~\ref{fig:angleProp}.
The basic idea of the following construction is to separate normal from tangential difference and to argue independently on each.
We now subdivide every triangle $\reference{T}_l$ of $\reference{S}$ (analogously for ${T}_l$ of $S$) into three parts $\reference{T}_{l_0}, \reference{T}_{l_1}, \reference{T}_{l_2}$ by means of the incenter and the bisecting lines of the angles. 
Without loss of generality we can assume that $\reference{T}_{j_0}$ is neighboring $\reference{T}_{i_0}$ (thus $\reference{T}_{j}$ was already neighboring $\reference{T}_{i}$). We fix a frame $\reference{F}_{i_0}$ on $\reference{T}_{i_0}$ ensuring alignment of the first basis vector to the edge shared with $\reference{T}_{j_0}$. Frame $\reference{F}_{j_0}$ is now defined by rotating $\reference{F}_{i_0}$ around the common edge. This directly implies that $\reference{C}_{i_{0}j_{0}}$ realizes an angle with absolute value smaller than $\nicefrac{\pi}{2}$. The same holds for $C_{i_{0}j_{0}}$ since $F_l=R_l\reference{F}_l$ preserves alignment of the frames with the underlaying triangles and through the initial refinement we already ensured normal differing of less then $\pm\pi$.
We analogously define $\reference{F}_{i_k}$ if $\reference{T}_{i_k}$ has neighbors, if not, we simply set $\reference{F}_{i_k}=\reference{F}_{i_0}$.
This construction allows to explicitly differentiate two different types of relative transition rotations: type (I) that comes from normal differences like ${C}_{i_{0}j_{0}}\reference{C}^{-1}_{i_{0}j_{0}}$ and type (II) like ${C}_{j_{0}j_{1}}\reference{C}^{-1}_{j_{0}j_{1}}$ that is induced by tangential change.
Since $F_{j_0}=R_j\reference{F}_{j_0}$ and $F_{j_1}=R_j\reference{F}_{j_1}$ we see immediately that ${C}_{j_{0}j_{1}}=\reference{F}^{-1}_{j_0}R^{-1}_jR_j\reference{F}_{j_1}=\reference{C}_{j_{0}j_{1}}$ and thus ${C}_{j_{0}j_{1}}\reference{C}^{-1}_{j_{0}j_{1}}=I_3$.
To clarify (I) we strip some notation, more precisely, let $C^{1}, C^{2}$ be two transition rotations with angles $-\theta^{1}, \theta^{2}$ and axes $-v^{1}, v^{2}$, respectively realizing normal change only.
Then, the relative transition rotation is given by $C^{{12}}=C^{2} \cdot \left(C^{1}\right)^{-1}$ and of type (I).
Now, assuming $\theta^{1}, \theta^{2}\in \left(-\nicefrac{\pi}{2}, \nicefrac{\pi}{2}\right)$ and in light of
\begin{align*}
\cos\left(\dfrac{\theta^{{12}}}{2}\right)&=\cos\left(\dfrac{\theta^{1}}{2}\right)\cos\left(\dfrac{\theta^{2}}{2}\right) - \sin\left(\dfrac{\theta^{1}}{2}\right)\sin\left(\dfrac{\theta^{2}}{2}\right)\scal{v^{1}}{v^{2}}
\end{align*}
(cf.\ e.g.~\cite{altmann2005rotations}), it follows that the angle $\theta^{{12}}$ of the composite rotation does not exceed $(-\pi, \pi)$, hence is well-localized.
\end{proof}
}

\subsection{Classification with varying commensuration parameter}
\setcounter{figure}{0}
\label{appx:varyCommPara}
\begin{figure*}[ht]
 \fboxsep0pt 
  \subfloat[\revisedII{OAI data: As $\omega$ increases accuracy does and compactness decreases.}\label{fig:varyOmegaOAI}]{
	\begin{minipage}[c]{0.49\linewidth}
	   \centering
	   \includegraphics[width=1\textwidth]{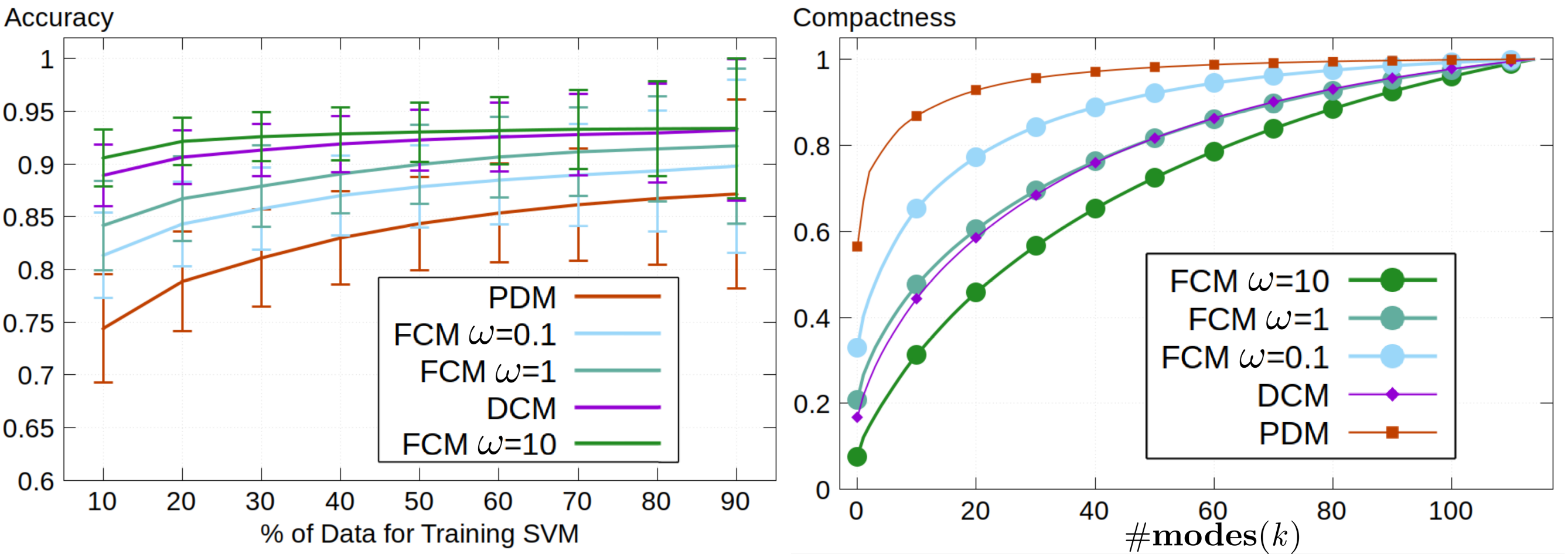}
	\end{minipage}}
 \hfill 	
  \subfloat[\revisedII{ADNI data: As $\omega$ increases compactness decreases, peak accuracy is obtained for $\omega\approx 0.98$.}\label{fig:varyOmegaADNI}]{
	\begin{minipage}[c]{0.49\linewidth}
	   \centering
	   \includegraphics[width=1\textwidth]{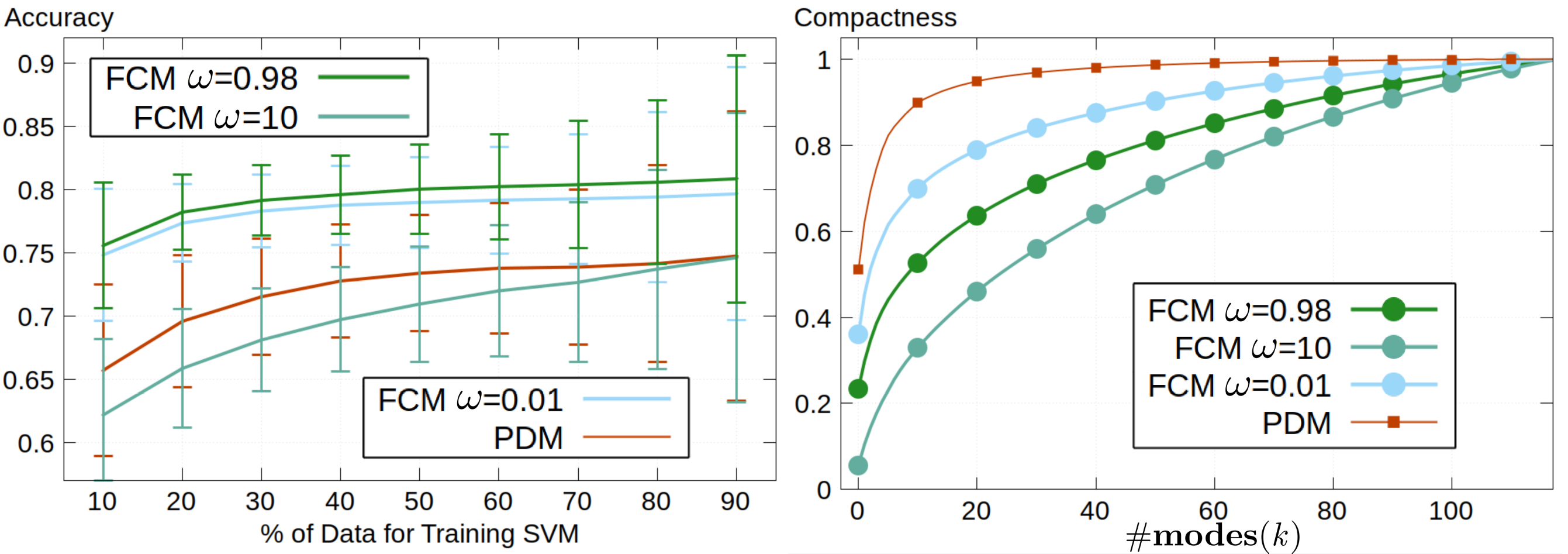}
	\end{minipage}}
\caption{\revisedII{Classification accuracy and compactness for varying commensuration parameter $\omega$.}}
\end{figure*}
\revised{As we are on the one hand applying shape models for disease classification purposes and on the other hand are in general interested in rather compact models, we did vary the metric commensuration parameter $\omega$ (cf.~Sec.~\ref{sec:discretization}) since it directly affects both. We studied the connection between choice of $\omega$, classification accuracy and model compactness.\par\noindent
{\bf OAI - OA Classification.}
As can be seen in Fig.~\subref{fig:varyOmegaOAI} the disease classification accuracy increases as $\omega$ increases. Looking at Eq.~\ref{eq:distance} this means putting higher weight on the rotational, thus curvature related term leads to higher classification accuracy. Additionally all examined choices of $\omega$ give FCM classification results with a superior performance compared to PDM.}
\revised{However, the development of model compactness is contrary to the classification accuracy as shown in Fig.~\subref{fig:varyOmegaOAI}. The larger $\omega$ gets, the less compact is the shape model and none of the examined commensuration parameter choices leads to a compactness as high as for the PDM.}
~\par
\noindent
\revised{{\bf ADNI - Alzheimer's Classification.}
A similar experiment for Alzheimer's classification reveals a rather different dependency on the commensuration parameter, see Fig.~\subref{fig:varyOmegaADNI}.
For values $\omega  \gtrsim 10$ the classification accuracy lies below the one achieved by the PDM. For $\omega \approx 0.98$ the peak performance is reached and for values below we note again slight decrease in performance.}
\revised{The compactness instead, as can be seen in Fig.~\subref{fig:varyOmegaADNI}, develops very similar as for the OAI dataset and is still for all $\omega$ below that of the PDM.}
\revised{As conclusion to this section we conjecture, that comparison of compactness might be interesting for models that are build on the same shape representation but it becomes less meaningful if different representations are compared. Furthermore, we find that the most compact models do not (necessarily) give the best classification accuracy. It appears that complex shape variation as it emerges from certain diseases tends to require a less compact encoding for an expressive but specific description.}

\subsection{Data Identifiers}
\setcounter{figure}{0} 
\setcounter{table}{0} 
\label{app:DataIds}
The data used within the given classification experiments is publicly available, we thus aim to facilitate reproduction of and comparison to our results. To this end we compiled labeled identifier lists of the data used in our experiments.
Regarding the OA classification all cases can be found as segmentation masks accompanying the publicly available \emph{OAI-ZIB} dataset\footnote{doi.org/10.12752/4.ATEZ.1.0}, whereas the Alzheimer's classification experiment relies on hippocampus segmentation masks that can be accessed as part of ADNI database\footnote{adni.loni.usc.edu}.
\begin{table}[tbh]
\label{tab:OAIids}
\centering
\caption{List of unique patient ids from the OAI database used in the OA classification experiment.}
{\footnotesize 
\begin{tabular}{lllcrrr}     \hline
        \multicolumn{3}{c}{Healthy (KL 0/1)}   &  & \multicolumn{3}{c}{Diseased (KL 4)}  \\ \hline
        9008561 & 9258563 & 9510418            &  & 9246518 & 9391984 & 9631713          \\ 
        9013798 & 9304351 & 9517914            &  & 9256759 & 9393987 & 9638953          \\ 
        9017909 & 9331053 & 9582487            &  & 9263504 & 9413071 & 9642550          \\ 
        9036770 & 9333574 & 9601162            &  & 9266394 & 9414291 & 9660708          \\ 
        9036948 & 9341699 & 9617689            &  & 9267719 & 9421492 & 9672573          \\ 
        9039744 & 9341903 & 9645577            &  & 9271965 & 9422381 & 9680800          \\ 
        9089627 & 9355112 & 9655592            &  & 9284505 & 9430102 & 9689922          \\ 
        9108461 & 9383004 & 9718992            &  & 9287216 & 9439428 & 9691663          \\ 
        9116298 & 9391372 & 9750072            &  & 9301332 & 9457359 & 9695135          \\ 
        9120941 & 9394136 & 9854269            &  & 9326657 & 9467278 & 9700341          \\ 
        9132486 & 9397088 & 9876530            &  & 9331465 & 9469318 & 9710479          \\ 
        9141244 & 9397976 & 9878765            &  & 9340139 & 9470313 & 9745458          \\ 
        9153509 & 9433408 & 9879069            &  & 9349261 & 9475286 & 9750090          \\ 
        9171766 & 9440417 & 9907090            &  & 9364366 & 9477175 & 9760079          \\ 
        9184495 & 9460287 & 9916140            &  & 9365968 & 9477358 & 9781749          \\ 
        9189553 & 9474901 & 9967815            &  & 9375317 & 9508335 & 9858216          \\ 
        9207016 & 9486748 & 9973322            &  & 9379276 & 9517311 & 9895555          \\ 
        9211049 & 9488834 & 9978579            &  & 9389580 & 9557454 & 9933836          \\ 
        9245519 & 9501871 & 9988421            &  & 9391061 & 9568504 & 9943227          \\ 
                & 9504627 &                    &  &         & 9604541 &                  \\ \hline
\end{tabular}
}
\end{table}
\begin{table}[tbh]
\label{tab:ADNIids}
\centering
\caption{List of unique scan ids from the ADNI database used in the Alzheimer's classification experiment.}
{\footnotesize 
\begin{tabular}{lllcrrr}     \hline
        \multicolumn{3}{c}{Cognitive Normal}   &  & \multicolumn{3}{c}{Alzheimer's Diagnosed}  \\ \hline
        10312 & 13681 & 17207            &  & 10064 & 14974 & 22310          \\ 
        10605 & 13717 & 17232            &  & 10468 & 15001 & 22938          \\ 
        10813 & 13737 & 17487            &  & 10568 & 15145 & 23375          \\ 
        10835 & 13893 & 17527            &  & 10764 & 15287 & 23446          \\ 
        10883 & 14104 & 18109            &  & 11633 & 15315 & 24659          \\ 
        10960 & 14488 & 18236            &  & 12000 & 15935 & 24672          \\ 
        11006 & 14513 & 18321            &  & 12365 & 16313 & 25082          \\ 
        11161 & 14559 & 18450            &  & 12375 & 16924 & 25357          \\ 
        11314 & 14818 & 18827            &  & 12381 & 17191 & 25455          \\ 
        11584 & 14959 & 18909            &  & 12402 & 17337 & 25763          \\ 
        11594 & 14991 & 18917            &  & 12468 & 18077 & 26038          \\ 
        11928 & 15079 & 19971            &  & 12583 & 18094 & 26136          \\ 
        11974 & 15527 & 20352            &  & 12836 & 18151 & 26143          \\ 
        12081 & 15727 & 20753            &  & 12952 & 18189 & 26314          \\ 
        12419 & 15789 & 21817            &  & 13839 & 18373 & 26431          \\ 
        12485 & 16048 & 22439            &  & 13976 & 18390 & 27061          \\ 
        12563 & 16099 & 24338            &  & 13990 & 19296 & 27414          \\ 
        12992 & 16553 & 25680            &  & 14199 & 19386 & 27584          \\ 
        13191 & 16759 & 25829            &  & 14629 & 19395 & 27673          \\ 
        13556 & 17131 & 26899            &  & 14699 & 21207 & 28133          \\ \hline
\end{tabular}
}
\end{table}

\bibliography{ms}

\end{document}